\documentclass[journal,12pt,onecolumn]{IEEEtran}

\usepackage{amssymb}
\usepackage{epsfig}
\usepackage{fullpage}
\usepackage{amsmath}
\usepackage{algpseudocode}
\usepackage{algorithm}
\usepackage{makecell}
\usepackage{subfigure}
\usepackage{fullpage}
\usepackage{graphicx}

\begin{document}


\title{Waddling Random Walk: Fast and Accurate Mining of Motif Statistics in Large Graphs}
\author{\IEEEauthorblockN{Guyue Han and Harish Sethu}\\
\IEEEauthorblockA{Department of Electrical and Computer Engineering\\
Drexel University\\
Philadelphia, PA 19104-2875\\
Email: \{guyue.han, sethu\}@drexel.edu}
}

\maketitle
\thispagestyle{empty}
~\vskip 0.5in
\begin{abstract}
Algorithms for mining very large graphs, such as those representing online social networks, to discover the relative frequency of small subgraphs within them are of high interest to sociologists, computer scientists and marketeers alike. However, the computation of these network motif statistics via naive enumeration is
infeasible for 
either its prohibitive computational costs or access
restrictions on the full graph data. Methods to estimate the motif
statistics based on random walks by sampling only a small fraction of
the subgraphs in the large graph address both of these challenges. In
this paper, we present a new algorithm, called the {\em Waddling Random
Walk (WRW)}, which estimates the concentration of motifs of any size. It derives its name from the fact that it sways a little
to the left and to the right, thus also sampling nodes not directly on
the path of the random walk. The WRW algorithm achieves its
computational efficiency by not trying to enumerate subgraphs around
the random walk but instead using a randomized protocol to sample
subgraphs in the neighborhood of the nodes visited by the walk. In
addition, WRW achieves {\em significantly} higher accuracy (measured by the
closeness of its estimate to the correct value) and higher precision
(measured by the low variance in its estimations) than the current
state-of-the-art algorithms for mining subgraph statistics. We illustrate these advantages in speed,
accuracy and precision using simulations on well-known and widely used
graph datasets representing real networks. 
\end{abstract}

%
%

%
%







\newpage
\section{Introduction}
The analysis of large graphs, such as those representing online
social networks, is of increasing scholarly interest to
sociologists, mathematicians, economists, computer scientists and
marketeers \cite{TirHal2015}. In particular, mining of large graphs for their microstructure describing patterns of relationships between neighboring
vertices, is of significant interest to researchers in data mining
\cite{ChuChe2011,EleSha2015,ZouHol2010,SilMei2010}. This
microstructure is best captured by motif or graphlet statistics, i.e.,
the relative frequencies with which different small subgraphs of a
certain size appear in the large graph
\cite{MilShe2002,BhuRah2012,UgaBac2013,JhaSes2015,PinJin2015,SahHas2015,WanLui2014,WanLui2015}. For
example, 
the clustering coefficient (the number 
of triangles in relation to the number of wedges) has long
served as an important metric in sociometry and social 
network analysis \cite{ChaFal2006,HarKat2013}. In fact, the relative
frequencies of network motifs are indicative of important properties
of graphs such as modularity, the tendency of nodes in a network to
form tightly interconnected communities, and even play a role in the
organization and evolution of networks
\cite{MilShe2002}. Knowledge of these motif statistics
combined with homophily, the tendency of similar nodes to connect to
one another, add to the ability of businesses such as Facebook to
better mine their graphs and monetize their social platforms through targeted advertisements
\cite{WanCha2013}. 

Computing motif statistics, however, is rendered difficult by two
challenges: one computational and the other having to do with
restricted access to the full graph data. The computational challenge
arises because accurate computation of the relative frequencies of
different motifs requires enumeration of all the induced subgraphs and
checking each for isomorphism to known motif types. The time
complexity of enumerating all induced subgraphs of size $k$ in a graph
with $V$ vertices and $E$ edges is exponential in $k$ with an upper bound of
$O(E^k)$ and a lower bound of $O(Vc^{k-1})$ \cite{ItzMog2007}. Even when $k$ is as small as $4$, in a
graph with only millions of edges, the number of motifs can reach
hundreds of billions. The other problem is one of restricted access because the data on many
large graphs, especially online social networks, can only be obtained
piecemeal via the platform's public interface encapsulated in its API
for developers on the platform. One common query allowed by most
social network APIs is one that returns the list of neighbors of a
node --- a feature that allows random walks on these large graphs even
when the full graph data is unavailable \cite{GjoKur2010}.

The computational and the access challenges above motivate the need
for an approach to estimating motif statistics via sampling the graph
using a random walk and checking only a small fraction of all the
induced subgraphs for isomorphism
\cite{GjoKur2010,HarKat2013,RibTow2010,WanLui2014}.

\subsection{Problem statement}

Consider a connected, undirected graph $G=(V,E)$ with vertex set $V$
and edge set $E$. We assume that information about the graph can only be
ascertained through querying each node separately for a list of its
neighbors. 

For convenience and clarity, we denote each motif by a unique 2-tuple,
$M(k, m)$, where $k$ is the number of vertices in the motif and $m$ is
the motif id which uniquely identifies a motif given $k$. Fig.~\ref{fig:motifs} illustrates all
motifs with $k\leq 5$. 

Let ${\bf S}(k)$ denote the set of all connected induced subgraphs
with $k$ vertices in $G$. Similarly, let ${\bf S}(k,m)$ denote the set
of all connected induced subgraphs which are isomorphic to motif
$M(k,m)$. Now, the motif statistics or motif concentrations are
given by the relative frequencies of each of the motif types:
\[
C(k,m) = \frac{|{\bf S}(k,m)|}{|{\bf S}(k)|}
\]

Given a large graph $G$, the problem considered in this paper is one 
of determining $C(k,m)$ for any $k$ and $m$ by visiting nodes
in the graph only through its public interface via a random walk. The
goal is to make an estimate that is accurate and precise while
visiting as few nodes as possible.  

\subsection{Related Work}

\begin{figure}[!t]
\begin{center}
\includegraphics[width=5in]{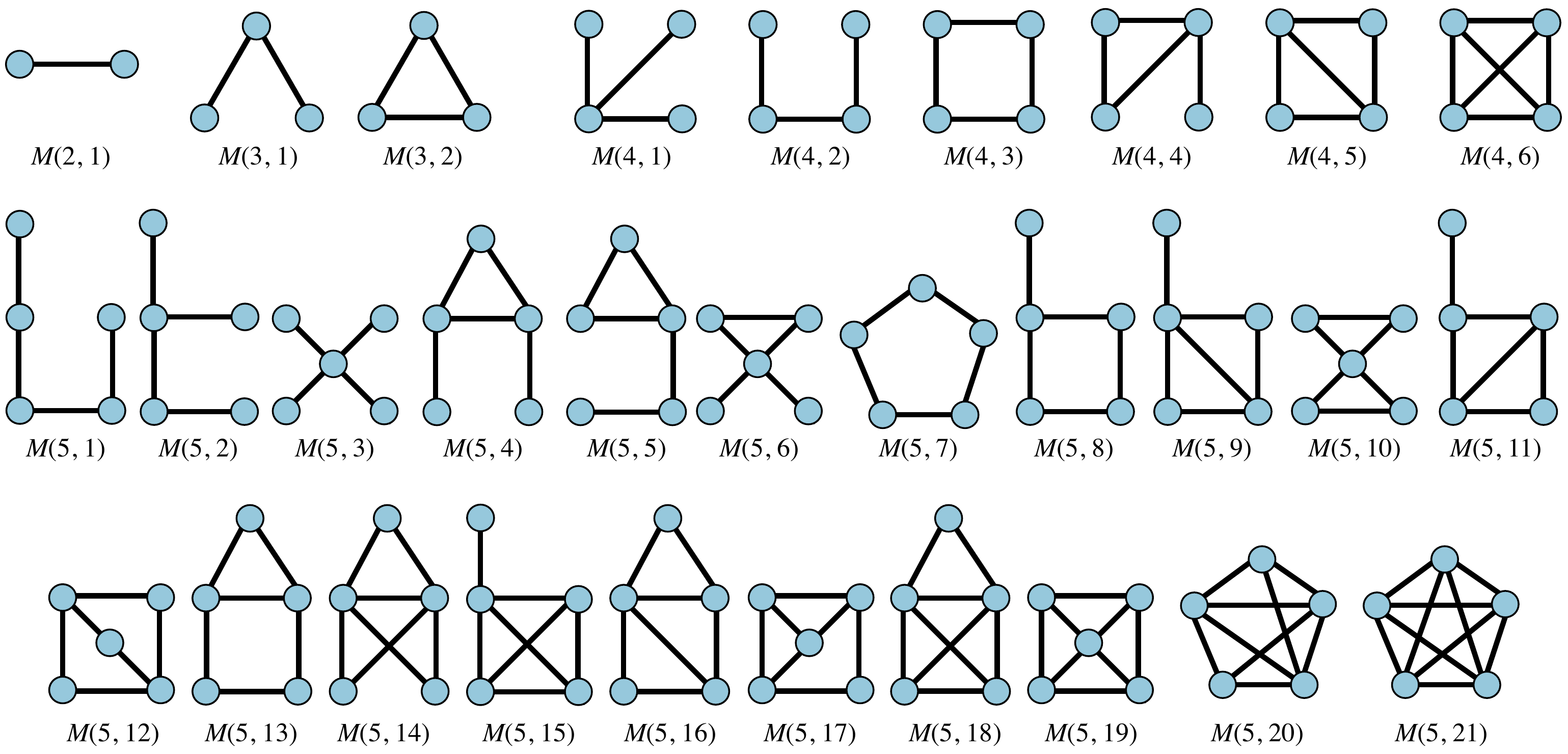} 
\caption{All 2, 3, 4 and 5 vertices undirected motifs.}
\label{fig:motifs}
\end{center}
\end{figure}  

The earliest work on motifs in large graphs began with studies of
triadic properties such as triangle counts and the global clustering
coefficient \cite{WasFau1994,ChaFal2006}. Since then, a large body of
work has focused on understanding and estimating the properties of
graphs related to $3$-node motifs. Yet, computing the accurate
statistics of even these smallest of motifs (wedges and triangles) is
prohibitively expensive for large graphs, inspiring multiple
efforts based on making estimates using edge sampling
\cite{TsoKan2009,AhmDuf2014,JhaSes2015space}.  
The class of approaches based on random walks, however, solve not only 
the computational challenge but also the typical restrictions imposed
on full access to the graph --- they allow piecemeal collection of
data by walking the graph querying a node at a time for its list of
neighbors \cite{GjoKur2010,HarKat2013}. 

A smaller but increasing body of work has tried to develop graph sampling methods that apply to motifs of size larger than three
\cite{Wer2006,BhuRah2012,RahBhu2012,WanLui2015,JhaSes2015}. Applications in
bioinformatics, in particular, have inspired these efforts due to the
need for motif detection and motif-related computations in biology
\cite{KimLi2011,PanRom2013}. There have been at least two
classes of approaches in the estimation of the statistics of
larger-size motifs: one based on edge sampling and the other based on
random walks. Edge sampling approaches are able to reduce the
computational complexity of making an estimation, but they usually
require knowledge of global properties of the graph (such as the total
number of edges in the graph) or they require access to the full
graph. Only methods based on random walks rely entirely on the public interface of live
networks and are able to address both the computational and the access
challenges mentioned in the previous section. 

One approach based on random walks uses the Metropolis-Hastings method \cite{StuRej2009} which can collect uniformly random nodes to infer motif
statistics. However, since nodes selected uniformly randomly may not necessarily induce connected subgraphs, a better approach is to build a graph
of connected induced subgraphs (CIS) and conduct a random walk on this
CIS graph \cite{WanLui2014,BhuRah2012,SahHas2015}. Two subgraphs in
this CIS graph are directly connected by an edge if they differ in
only one node in the original graph. Starting from one 
subgraph, one can move to a neighboring subgraph by dropping and
adding a node without having pre-computed the entire graph of
subgraphs. This approach to a subgraph random walk is improved in
\cite{BhuRah2012} using a Metropolis-Hastings based sampling method
to perform a uniform sampling of CISs in the large graph, leading to a
Markov Chain Monte Carlo sampling method for estimating the motif
frequency distribution of 3-node, 4-node and 5-node motifs. The use of
Metropolis-Hastings for walking the CIS graph is further refined in
\cite{SahHas2015} to collect motif statistics, in an algorithm called
the Metropolis-Hastings Random Walk (MHRW).

An alternative approach, also based in random walks on CIS graphs, is
one that avoids the Metropolis-Hastings method for its inefficiency
involving randomized selections and the consequent rejections of nodes in
determining the next step in the walk. Instead, in this approach, the
unbiased sampling of Metropolis-Hastings method is replaced with the
use of the Horvitz-Thompson construction to unbias the estimation
\cite{HorTho1952}. Such a method is used in \cite{WanLui2014} which
develops the Pairwise Subgraph Random Walk (PSRW), which cleverly
samples a set of CISs with a smaller number of $k-1$ nodes to estimate
the concentrations of motifs with $k$ nodes.

Both PSRW and MHRW are capable of estimating motif concentrations of
any size. As presented in \cite{WanLui2014,SahHas2015}, these two
algorithms are significantly better than the existing methods in terms
of accuracy and speed. However, both of these algorithms rely on some
subgraph enumeration which adds significantly to the runtime. The
Waddling Random Walk (WRW), proposed in this paper, however, avoids
such enumeration and instead uses a randomized approach to sample
subgraphs and reduce computational costs. WRW achieves a significant
improvement in speed as well as in the accuracy and the precision of
its estimates.

\subsection{Contributions}

We present a new random walk algorithm, called {\em Waddling Random
  Walk (WRW)}, named so because it sways left and right and also
samples nodes not directly in the path of the random
walk. In Section~\ref{sec:rationale}, we develop the theoretical
foundation for the algorithm and show that motif statistics can be
inferred from the probability with which we sample sets of nodes and
whether or not the subgraphs induced by those nodes are isomorphic to
the motifs of interest.

Section~\ref{sec:wrw} presents the WRW algorithm to sample $k$-node
motif statistics for any $k$ along with a pseudocode description of
it. The algorithm relies on a randomized waddling protocol to sample
nodes in the neighborhood of the random walk. A key strength of the
algorithm is that the waddling protocol can be customized for specific
access or other constraints, with the only requirement being that it
be a {\em randomized} protocol so that the probabilities of sets of
nodes selected by the protocol can be computed. Section~\ref{sec:wrw}
also describes the specific version of the algorithm for collecting
$4$-node and 5-node motif statistics. 

Section~\ref{sec:results} describes a thorough performance analysis of
WRW in comparison to the best two algorithms which address the same
problem: PSRW introduced in 2014 \cite{WanLui2014} and MHRW introduced
in 2015 \cite{SahHas2015}. We show that WRW achieves a significantly
improved running time. Most importantly, we show, using graph datasets
representing real networks, that the WRW algorithm achieves
significantly higher accuracy (in terms of the closeness of its answers to
the actual values) and higher precision (in terms of the variance in
its estimations). We also show that WRW can estimate the number of motifs of any type if the size of the network is known or is estimated.

Section~\ref{sec:conclusion} concludes the paper.

\section{The Rationale}
\label{sec:rationale}

In this section, we build the theoretical rationale for the Waddling
Random Walk. In particular, we illustrate the need for waddling during
the random walk by first considering a simpler algorithm without waddling.

\subsection{Preliminaries and Notation}

Given a graph $G=(V,E)$, let $v \in V$ denote a vertex in $G$ and let
$N(v)$ denote the set of neighbors of vertex $v$ in $G$. Let $d(v)$
denote the degree of vertex $v$ and let $D=\sum_{v\in V}d(v)$ denote the sum of the
degrees of all the vertices in $G$. 

Consider the $k$-node motif $M(k,m)$. Let $l(k,m)$
denote the number of vertices in the shortest path (allowing repeated
vertices) in motif $M(k,m)$ that includes all of the motif's $k$
vertices. For example, $l(4,1)$ is $5$ while $l(4,2)$ is $4$. 

A path is called simple if it does not have any repeated vertices. Let
$L(k,m)$ denote the number of vertices in the longest simple path of
motif $M(k,m)$. For example, $L(4,1)$ is $3$ while $L(4,2)$ is $4$.

Let $T_k$ denote the number of different $k$-node motifs. For example,
$T_3 = 2$, $T_4 = 6$ and $T_5 = 21$.

Let $P_r(k,m,s)$ denote the number of different paths with $s$ vertices
(allowing repeats) in motif $M(k,m)$ which include all of the $k$
nodes. For example, $P_r(3,1,3)$ is $2$ while $P_r(3,2,3)$ is
$6$. Similarly, $P_r(4,1,5)$ is $6$, $P_r(4,2,4)$ is $2$ while $P_r(4,6,4)$
is $24$.

Consider a random walk on $G$, $(r_1, r_2, \dots )$, where $r_1$
denotes the starting node and $r_i$ denotes the node visited in step
$i$. Let $t$ denote the number of steps in the random walk required to
reach the mixing time \cite{LovWin1998}, i.e., when the
probability of visiting a given node in a given step reaches a
stationary distribution and is largely independent of the initial node
$r_1$ chosen to begin the random walk. In many real networks,
including social networks in particular, $t$ is small and usually of
the order of a few hundreds of nodes \cite{HarKat2013,StuRej2009}.

Let $\phi_i(v_j)$ denote the probability that the random walk visits
node $v_j$ in step $i$. For $i > t$, the mixing time, we can
drop $i$ from the notation and denote by $\phi(v_j)$ the probability
that the random walk visits node $v_j$ in any given step. In the rest
of this paper, we assume that all the computations are based on
observations made in the random walk after the mixing time is
reached. As shown in \cite{lovasz1993random}, in a random walk, 
$\phi(v_j)$ is given by:
\begin{equation}
\phi(v_j) = \frac{d(v_j)}{D}
\end{equation}

Let ${\bf R}^{(s)}$ denote the set of all sequences of $s$ nodes which
may appear in a random walk in $G$; it is the set of all $s$-node
paths (allowing revisits to nodes) in $G$. Let $X^{(s)} = (x_1, x_2,
\dots, x_s)$ represent a sequence of $s$ nodes such that $X^{(s)} \in {\bf R}^{(s)}$.
At any given point in the random walk, let $\phi(X^{(s)})$ denote the
probability that it steps through exactly the sequence of nodes $X^{(s)}$. Then, $\phi(X^{(s)})$
is given by:
\begin{eqnarray}
\nonumber \phi(X^{(s)}) & = &
                                         \frac{d(x_1)}{D}\frac{1}{d(x_1)}\frac{1}{d(x_{2})}
                               \cdots \frac{1}{d(x_{s-1})}\\ 
& = & \frac{1}{D}\frac{1}{d(x_2)} \cdots \frac{1}{d(x_{s-1})} \label{eqn:selectionProb}
\end{eqnarray}

Let $H(X^{(s)})$ denote the subgraph in $G$ induced by the set of
vertices in the random walk sequence $X^{(s)}$. If the number of distinct nodes in $X^{(s)}$ is $k$, then
$H(X^{(s)})$ is isomorphic to one of the $k$-node motifs. Define the
function $\omega(X^{(s)},k,m)$ as follows to indicate if $H(X^{(s)})$ is
isomorphic to motif $M(k,m)$:
\begin{equation*}
\omega(X^{(s)},k,m)=\left\{
\begin{array}{ll}
  1   &    \mbox{if $H(X^{(s)})$ is isomorphic to $M(k,m)$,}\\
  0   &    \mbox{otherwise.}
\end{array} 
\right. 
\end{equation*}
Note that the function $\omega(X^{(s)},k,m)$ does not depend on the order
of the nodes in the sequence $X^{(s)}$.

As we traverse nodes in the random walk, we can observe the sequences
of nodes visited and compute the probability that those sequences are
encountered using the expression in Eqn.~(\ref{eqn:selectionProb})
above. Then, we can evaluate $\omega(X^{(s)},k,m)$ for those sequences
to check if they induce a subgraph isomorphic to a certain
motif. The rest of this section explains how we can infer motif
statistics from these quantities.

\subsection{Motif statistics without waddling}

Given a motif $M(k,m)$ which appears in $G$, there are
$P_r(k,m,s)$ ways in which an $s$-node path may traverse this motif
visiting all its nodes. 
Therefore, if we sum up the function $\omega(X^{(s)}, k,
m)$ for every path $X^{(s)} \in {\bf R}^{(s)}$, we should get the
total number of motifs of type $M(k,m)$ multiplied by $P_r(k,m,s)$. More
formally,  
\begin{align}
{\displaystyle \sum_{X^{(s)} \in {\bf R}^{(s)}}} \omega(X^{(s)}, k, m) = P_r(k,m,s)|{\bf S}(k,m)| \label{eqn:summation}
\end{align}

For any sequence of nodes on the random walk, $X^{(s)}$, define
$f(X^{(s)})$ as follows using Eqn.~(\ref{eqn:selectionProb}): 
\begin{equation}
f(X^{(s)}) = \frac{1}{\phi(X^{(s)})D} = d(x_2)d(x_3) \cdots d(x_{s-1}) \label{f-nowaddle}
\end{equation}

Let $R_i^{(s)} \in {\bf R}^{(s)}$ denote the sequence of $s$ nodes
visited during steps $i-s+1$ through $i$, i.e., $(r_{i-s+1},
r_{i-s+2}, \dots, r_i)$. Consider the expected value of
$\omega(R_i^{(s)}, k, m)f(R_i^{(s)})$ over the random walk:
\begin{eqnarray}
E\left[\omega(R_i^{(s)}, k, m)f(R_i^{(s)})\right] &=& \sum_{X^{(s)} \in {\bf R}^{(s)}} \phi(X^{(s)}) \left(\omega(X^{(s)}, k, m)f(X^{(s)})\right) \nonumber \\
& =& \left( \frac{1}{D} \right) \sum_{X^{(s)} \in {\bf R}^{(s)}} \omega(X^{(s)}, k, m) \label{eqn:expected-no-waddle}
\end{eqnarray}
Using Eqn.~(\ref{f-nowaddle}) for $f(R_i^{(s)})$ on the LHS and
substituting for the above summation using Eqn.~(\ref{eqn:summation})
on the RHS, we get:
\begin{align}
 E\left[\omega(R_i^{(s)}, k, m) \prod_{j=1}^{s-2}d(r_{i-j})\right]
 = \left( \frac{1}{D} \right) P_r(k,m,s)|{\bf S}(k,m)| \label{eqn:nowaddle}
\end{align}

Eqn.~(\ref{eqn:nowaddle}) suggests a simple algorithm for sampling
motif statistics via a random walk. For each motif type $M(k,m)$, let
$s = l(k,m)$, the number of vertices in the shortest path in the motif
that includes all of its $k$ vertices. As we visit nodes in the random
walk, we can check if the previous $s$ nodes induce a subgraph
isomorphic to $M(k,m)$ to evaluate $\omega(R_i^{(s)}, k,
m)$. Fig.~\ref{fig:nowaddle-example} illustrates motifs recognized by
such an algorithm during a random walk.
\begin{figure}[!t]
\begin{center}
\includegraphics[width=2.5in]{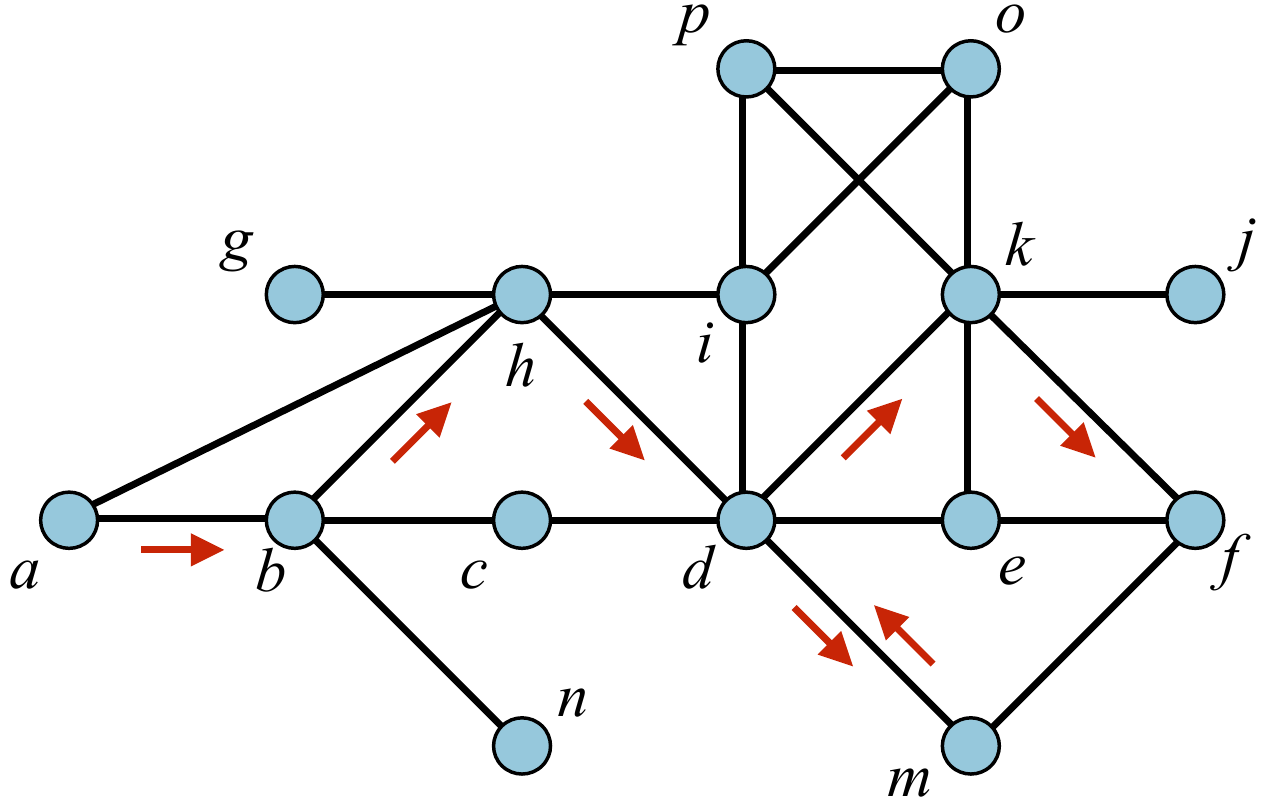}\\~\\
\includegraphics[width=3.2in]{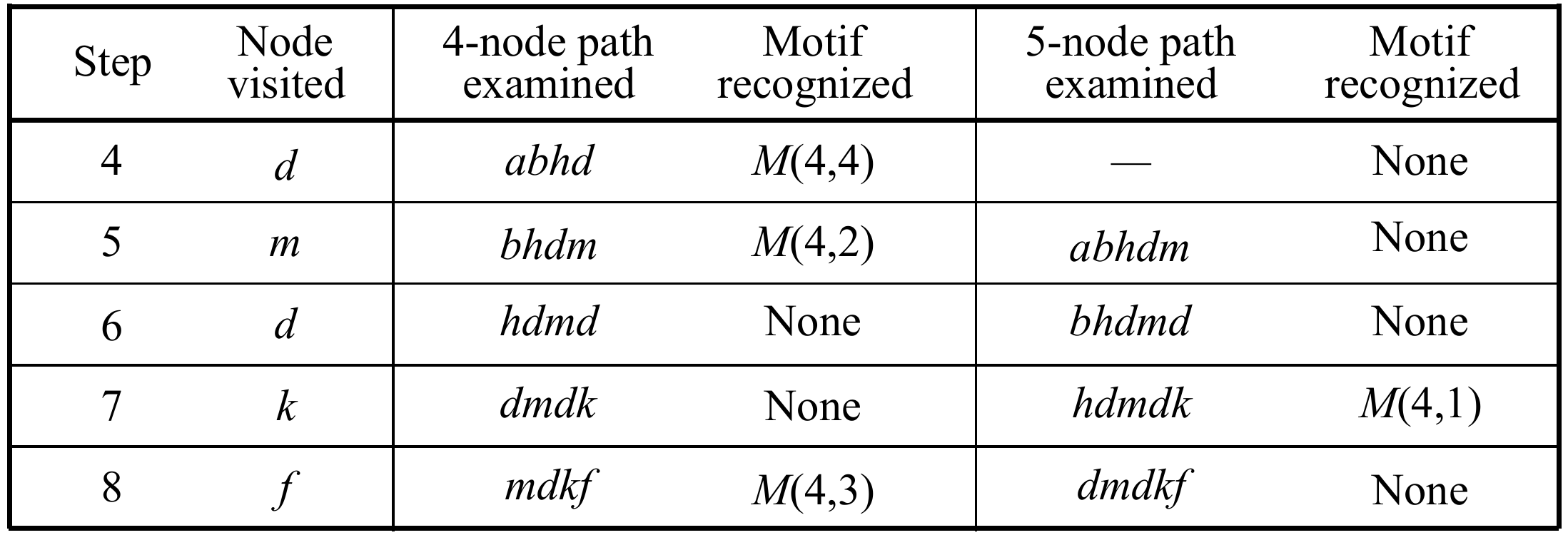} 
\caption{An example to illustrate the collection of $4$-node motif
  statistics in a random walk (shown by red arrows) without a
  waddle. $5$-node paths are examined only to check for motif $M(4,1)$.}
\label{fig:nowaddle-example}
\end{center}
\end{figure} 

At each step, if the isomorphism test passes, we can compute the
product of the degrees of the middle $s-2$ nodes in $R_i^{(s)}$ and
obtain the average of these results to get the LHS in
Eqn.~(\ref{eqn:nowaddle}) for $M(k,m)$, which we denote by
LHS$(k,m)$. In the RHS of Eqn.~(\ref{eqn:nowaddle}), since
$P_r(k,m,s)$ is known for all the motifs and $D$ is a constant, we can
compute the fraction of $k$-node motifs in a graph which are of a
certain type as follows:
\[
C(k,m) = \left( \frac{\mathrm{LHS}(k,m)}{P_r(k,m,s)} \right) \Bigg/ \left(
  {\displaystyle \sum_{j=1}^{T_k} \frac{\mathrm{LHS}(k,j)}{P_r(k,j,s)}} \right)
\]

\begin{table}[!t]
\centering
\caption{A glossary of selected basic terms.\label{table:term}}{
\begin{tabular}{lp{2.4in}}
Notation & Definition \\
\hline  \\[-8pt]
$d(v)$    &Degree of node $v$ in $G$.\\[1pt]
\hline \\[-8pt]
$D$    &Sum of the degrees of all nodes in $G$.\\[1pt]
\hline \\[-8pt]
$T_k$    &Number of different $k$-node motifs.\\[1pt]
\hline \\[-8pt]
$M(k,m)$    &$k$-node motif with id $m$.\\[1pt]
\hline \\[-8pt]
${\bf S}(k,m)$    &Set of all connected induced subgraphs which are isomorphic to $M(k,m)$.\\[1pt]
\hline \\[-8pt]
$l(k, m) $    &Number of vertices in the shortest path in $M(k,m)$ that includes all of its k vertices.\\[1pt]
\hline \\[-8pt]
$L(k,m) $    &Number of vertices in the longest simple path of motif $M(k,m)$.\\[1pt]
\hline \\[-8pt]
$P_r(k,m,s)$          & Number of different paths (allowing repeated
                        nodes) of length $s$ in motif $M(k,m)$\\[1pt]
\hline \\[-8pt]
$R_i^{(s)}$          & Sequence of $s$ nodes visited during steps $i-s+1$ through $i$.\\[1pt]
\hline \\[-8pt]
$H(R_i^{(s)})$          & Subgraph in $G$ induced by the set of vertices in the random walk sequence $R_i^{(s)}$.\\[1pt]
\hline \\[-8pt]
$C(k,m)$          & Motif concentration of $M(k,m)$.\\[1pt]
\hline \\[-8pt]
\end{tabular}}
\end{table}

\subsection{Why waddle?}

The algorithm suggested by Eqn.~(\ref{eqn:nowaddle}) in the previous
section works well when $l(k,m)$ is equal to $k$ but
can become less accurate when $l(k,m)$ is larger than $k$. For example, in
the case of the $4$-star motif or $M(4,1)$, $l(4,1)$ is $5$ and so the
random walk has to take $5$ steps within the motif to encounter and
recognize the motif; this means that motifs with larger $l(k,m)$ would
be encountered and recognized with lower probability, especially so in
large social network graphs with high average degree.

A significant improvement is possible if we allow our random walk to
waddle a little (sway left and right) and query random nodes to the
right and the left of the random walk as well. 
For example, consider the random walk illustrated in
Fig.~\ref{fig:nowaddle-example} to collect $4$-node motif
statistics. Suppose, in addition to the nodes visited on
the random walk, we also query a random neighboring node of each node
visited directly on the random walk. Suppose we query node {\em g}
at the step in which the walk visits node {\em h}. Then, when the walk
visits node {\em d} for the first time, we can recognize the $4$-star
motif $M(4,1)$ induced by nodes {\em h}, {\em b}, {\em g} and {\em d}
in addition to recognizing motif $M(4,4)$ in the same step
induced by nodes {\em a}, {\em b}, {\em h} and {\em d}.

Waddling, since it also examines nodes not in the direct path of the
random walk, allows us to restrict the number of previously visited
nodes along the walk that we examine for isomorphism to a motif
$M(k,m)$ to no more than the length in the number of nodes, $L(k,m)$, of
the longest simple path on the motif. Since $L(k,m) \leq k$, we
will sample motifs with a higher probability during every step of the
walk. Waddling helps count more motifs and thus improves the
accuracy of the motif statistics collected. As we will show in the
next section, for best efficiency, how we waddle (i.e., which other
nodes we query along the random walk and how deep a chain of nodes we
query) depends on the motif for which we are seeking to collect
statistics. But, as long as we can correctly compute the probability
of choosing the set of nodes for which we examine the induced
subgraphs, the methodology detailed in this section can be transferred
to the waddling algorithm to estimate the motif statistics.

\newcommand{\algrule}[1][.2pt]{\par\vskip.3\baselineskip\hrule height #1\par\vskip.3\baselineskip}
\begin{algorithm}[!t]
\caption{Waddling Random Walk}\label{alg:motifalg}
\begin{algorithmic}[1]
\Require{Graph $G = (V,E)$, motif size $k$, motif id $m$, random walk length $n$.}
\Ensure{Motif concentration $C(k,m)$}
\State{$c_m \leftarrow 0$, $1 \leq m \leq T_k$}
\State{Perform random walk until after the mixing time,
  reaching node $r_{i-1}$ at step $i-1$}
\While{$i < n$}
\State{$r_i \leftarrow $} Random node in $N(r_{i-1})$
\For{$m: 1, \dots, T_k$}
\State{$s = L(k,m)$}
\State{$R_i^{(s)} \leftarrow (r_{i-s+1}, \dots, r_i)$}
\If{Nodes in $R_i^{(s)}$ are all distinct}
\If{$s = k$}
\If{$H(R_i^{(s)})$ is isomorphic to $M(k,m)$}
\State{$c_m \leftarrow c_m + {\displaystyle \left(
      \prod_{j=1}^{s-2}d(r_{i-j}) \right) \Bigg/ P_r(k,m,s) }$}
\EndIf
\Else
\State{Pick a random $s$-node path of $M(k,m)$ and map it on the
  nodes in $R_i^{(s)}$}
\State{$W_i^{(k-s)} \leftarrow $ Set of $k-s$ nodes chosen by the
  randomized waddle protocol}
\If{$H(R_i^{(s)} \cup W_i^{(k-s)} )$ is isomorphic to $M(k,m)$}
\State{$c_m \leftarrow c_m + {\displaystyle \left( \frac{\displaystyle
        Z(k,m) \prod_{j=1}^{s-2}d(r_{i-j})}{\phi(W_i^{(k-s)} | R_i^{(s)})
        P_r(k,m,s) P_w(k,m,s) }\right)}$}
\EndIf 
\EndIf
\EndIf
\EndFor
\State{$i \leftarrow i+1$}
\EndWhile 
\State{$c_t \leftarrow {\displaystyle \sum_{j=1}^{T_k} c_j}$}
\State{\textbf{return} $c_m/c_t$}
\end{algorithmic}
\end{algorithm} 
\section{Waddling Random Walk}
\label{sec:wrw}

Algorithm~\ref{alg:motifalg} presents the pseudocode of the Waddling
Random Walk (WRW) to compute the concentrations of $k$-node motifs for
any $k$. In our algorithm, we use $T_k$ different temporary variables, $c_m$
for $1 \leq m \leq T_k$,  in which we record the $T_k$ motif
concentrations, one for each type. Lines 1--2 in the pseudocode
perform necessary initializations, begin the random walk and proceed
until the mixing time is reached.  

Lines 3--23 describe the walk after the mixing time, during
which period we collect the motif statistics. Line $4$ takes the next
step in the random walk to reach node $r_i$. At each step of the walk,
the {\tt for} loop in lines 5--21 loops through the processing 
required for each motif type --- the loop can be further optimized for
computational efficiency; we present the pseudocode as such for
clarity at the expense of some efficiency. Consider a 
motif $M(k,m)$ and let $s = L(k,m)$, the number of vertices in its
longest simple path. Note that $s \leq k$. Let $R_i^{(s)} =
(r_{i-s+1}, \dots, r_i)$ denote the sequence of $s$ nodes visited
during steps $i-s+1$ through $i$. 

Lines 8--20, expressed in generalized form, is the heart of the
algorithm and we will describe these at length. If the nodes in
$R_i^{(s)}$ are {\em not} all distinct, we will not recognize any
motifs of the type being considered and we will move forward to
either check for the next motif type or take the next step in the
random walk if all motif types at the current step have already been
considered.

If the nodes in $R_i^{(s)}$ are all distinct, there are two cases to
consider depending on the motif type: $s=k$ and $s < k$. If $s = k$,
then there is no need to waddle and we can use the approach in the previous section to
check for the isomorphism of $H(R_i^{(s)})$ and $M(k,m)$, and add to the
motif count toward estimation of the motif concentration
$C(k,m)$. Lines 9--12 handle this case when $s=k$.

In the other case when $s <k$, we can map the longest simple path of
the motif $M(k,m)$ onto these $R_i^{(s)}$ nodes (line 14). Note that there
are $P_r(k,m,s)$ different assignments that can accomplish the mapping
and we randomly choose one of them. To test for isomorphism to
$M(k,m)$, we now need at least an additional $k-s$ nodes. This is
accomplished by waddling in line $15$ described in greater
detail below.

Consider a randomized {\em waddle} protocol which queries an
additional set of $k-s$ choices of nodes, $W_i^{(k-s)}$, such that
$R_i^{(s)} \cup W_i^{(k-s)}$ induces a connected subgraph. The
querying of these nodes in $W_i^{(k-s)}$, which may be to the right or the left as
we take the random walk, produces the waddle for which the algorithm
is named. Note that the waddle protocol chooses the nodes
$W_i^{(k-s)}$ randomly and cannot guarantee if the induced subgraph
will be isomorphic to $M(k,m)$ for some $m$ or even if it has exactly
$k-s$ distinct nodes (due to the randomization, it may choose the same
node more than once). A powerful feature of our algorithm is that it
does not actually prescribe a specific waddle protocol --- for the
Waddling Random Walk to work, we only need the waddle protocol to be
randomized. In fact, the waddle protocol may be customized and
optimized for different motifs; we will provide the waddle protocol
optimized for $4$-node and $5$-node motifs later in this section.  

\begin{figure}[!b]
\begin{center}
\includegraphics[width=3.2in]{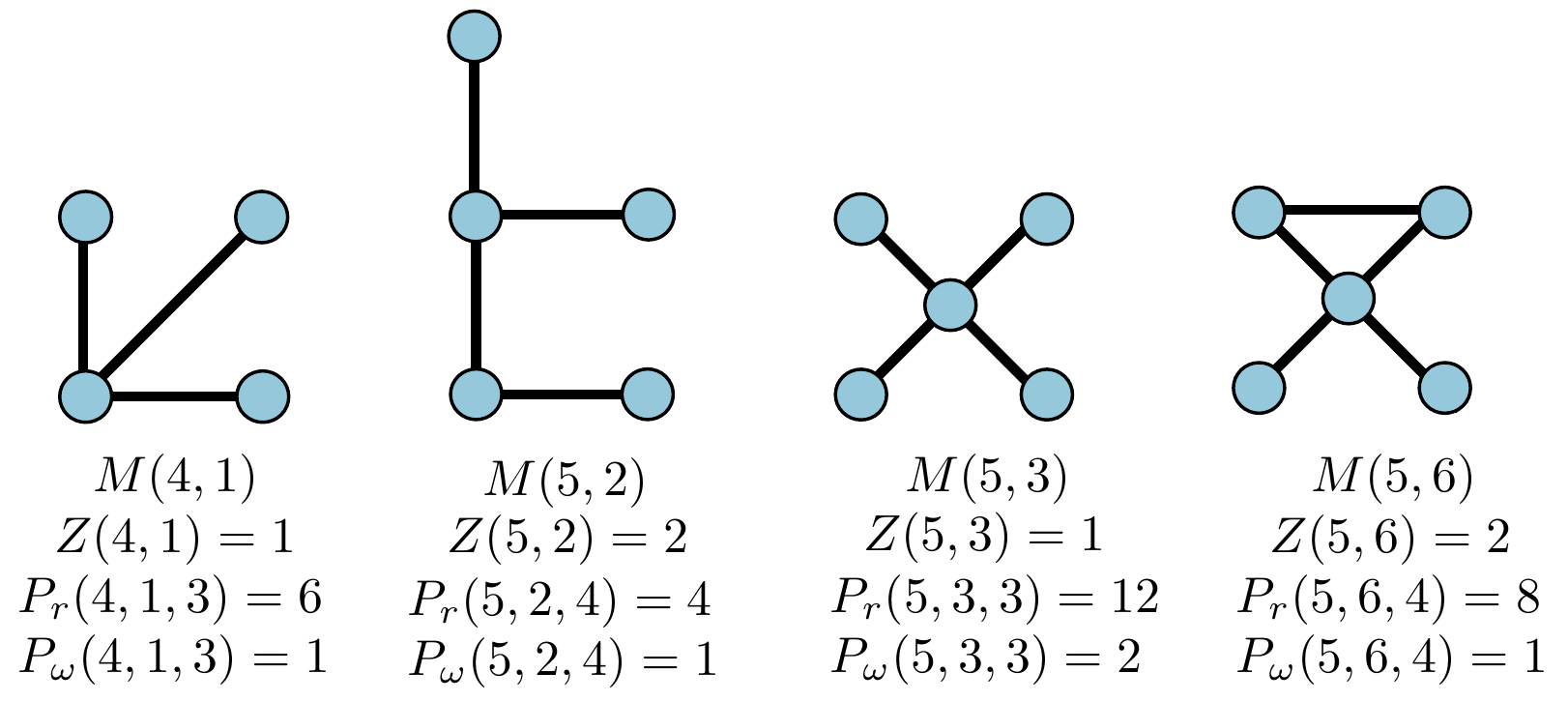}
\caption{$Z(k,m)$, $P_r(k,m,s)$ and $P_w(k,m,s)$ values for some motifs.}
\label{fig:P-example}
\end{center}
\end{figure}   

Let $\phi(W_i^{(k-s)} | R_i^{(s)})$ denote the probability that the 
randomized waddle protocol chooses exactly the nodes in
$W_i^{(k-s)}$ when the random walk visits the nodes in the 
sequence $R_i^{(s)}$. In the $4$-node motif sampling example
illustrated in Fig.~\ref{fig:nowaddle-example}, if $R_i^{(3)} = (b, h,
d)$ and if the waddle protocol randomly chooses a neighbor of $h$ to
include in $W_i^{(1)}$, then the probability $\phi(W_i^{(1)}
|R_i^{(3)} )$ is $1/d(h)$. 

Consider $\phi(R_i^{(s)}, W_i^{(k-s)})$, the probability that
the random walk visits the sequence of nodes $R_i^{(s)}$ and then the
waddle protocol chooses the set of nodes $W_i^{(k-s)}$ at step $i$ of
the random walk. Using Eqn.~(\ref{eqn:selectionProb}), it is given by: 
\begin{equation}
\phi(R_i^{(s)}, W_i^{(k-s)}) = \frac{\phi(W_i^{(k-s)} | R_i^{(s)})}{D\prod_{j=1}^{s-2} d(r_{i-j})} \label{eqn:selectionProbWaddle}
\end{equation}

Let $H(R_i^{(s)} \cup W_i^{(k-s)})$
denote the subgraph induced by the set of nodes in $R_i^{(s)} \cup W_i^{(k-s)}$. As in
the previous section, define the function $\omega(R_i^{(s)} \cup W_i^{(k-s)},k,m)$ as
follows to indicate if $H(R_i^{(s)} \cup W_i^{(k-s)})$ is isomorphic to motif $M(k,m)$: 
\begin{equation*}
\omega(R_i^{(s)} \cup W_i^{(k-s)},k,m)=\left\{ 
\begin{array}{ll}
  1   &    \mbox{if $H(R_i^{(s)} \cup W_i^{(k-s)})$ is isomorphic to $M(k,m)$},\\
  0   &    \mbox{otherwise.}
\end{array} 
\right. 
\end{equation*}

Let ${\bf W}^{(k-s)}$ be the set of all collections of $k-s$ nodes
(allowing repeated nodes) such that if $Y^{(k-s)} \in {\bf 
  W}^{(k-s)}$, then  for some $X^{(s)} \in {\bf R}^{(s)}$, $X^{(s)}
\cup Y^{(k-s)}$ induces a connected subgraph in $G$. 

Depending on the motif type, note that the node-to-node mapping of
$R_i^{(s)} \cup W_i^{(k-s)}$ to the nodes in the motif may not be
possible if the mapping of nodes in $R_i^{(s)}$ to the longest simple
path of the motif is reversed, i.e., if the node mapped to $r_i$ is now
mapped to $r_{i-s+1}$ and vice-versa and so on. Define $Z(k,m)$ as $1$
if the mapping is possible under such a reversal and $2$ otherwise;
note that $Z(k,m)$ is a property of the motif indicating if the motif
is lengthwise symmetric around the longest simple path.

Similarly as in the case of Eqn.~(\ref{eqn:summation}), 
\begin{eqnarray}
& {\displaystyle \sum_{X^{(s)} \in {\bf R}^{(s)}} ~ \sum_{Y^{(k-s)} \in {\bf
  W}^{(k-s)}}} \omega(X^{(s)} \cup Y^{(k-s)}, k, m) \nonumber \\
& ~~~~ = {\displaystyle \left( \frac{1}{D} \right) \left(
  \frac{P_r(k,m,s)P_w(k,m,s)|{\bf S}(k,m)|}{ Z(k,m) } \right)} \label{eqn:summation-waddle}
\end{eqnarray}
where $P_r(k,m,s)$ is the number of different paths of length $s$ in
motif $M(k,m)$, and $P_w(k,m,s)$ is the number of different ways in
which nodes in $Y^{(k-s)}$ can then be mapped on to the nodes not on
the $s$-node path in motif $M(k,m)$. 
While $P_r(k,m,s)$ captures the number of different ways in which
an $s$-node path can be mapped on to the longest simple path of the
motif, $P_w(k,m,s)$ captures the number of different ways one can map
the additional $k-s$ nodes chosen by the waddle protocol on to the
remaining $k-s$ nodes of the motif. Fig.~\ref{fig:P-example} shows
examples of some $k$-node motifs whose longest simple path is of
length $s < k$ and the corresponding values of $Z(k,m)$, $P_r(k,m,s)$ and
$P_w(k,m,s)$.

Using Eqn.~(\ref{eqn:selectionProbWaddle}), define a function $f(.)$
such that:
\[
f(R_i^{(s)}, W_i^{(k-s)}) = \frac{1}{\phi(R_i^{(s)}, W_i^{(k-s)})D}
= \frac{\prod_{j=1}^{s-2}d(r_{i-j})}{\phi(W_i^{(k-s)} | R_i^{(s)})} 
\]
Let $\Omega(k,m,X^{(s)},Y^{(k-s)})$ denote the function given by the
product:
\[
\phi(X^{(s)}, Y^{(k-s)}) \omega(X^{(s)}\cup Y^{(k-s)}, k, m)
  f(X^{(s)}, Y^{(k-s)})
\]

As in Eqn.~(\ref{eqn:expected-no-waddle}), we have the following
expected value:
\begin{align}
& E[\omega(R_i^{(s)} \cup W_i^{(k-s)}, k, m)f(R_i^{(s)},W_i^{(k-s)})] \nonumber \\
& ~~~~ = \sum_{X^{(s)}\in {\bf R}^{(s)}} ~ \sum_{Y^{(k-s)} \in {\bf W}^{(k-s)}}
  \Omega(k,m,X^{(s)},Y^{(k-s)})
\end{align}
Using Eqns. (\ref{eqn:selectionProbWaddle}) and
(\ref{eqn:summation-waddle}), we get:
\begin{align}
& E\left[\omega(R_i^{(s)} \cup W_i^{(s)}, k, m)
  \frac{\prod_{j=1}^{s-2}d(r_{i-j})}{\phi(W_i^{(k-s)} | R_i^{(s)})}
  \right] \nonumber \\
& ~~~~ = {\displaystyle \left( \frac{1}{D} \right) \left(
  \frac{P_r(k,m,s) P_w(k,m,s)|{\bf S}(k,m)|}{Z(k,m)} \right) } \label{eqn:final}
\end{align}

Fig.~\ref{fig:waddle-example} shows an example of how $4$-node motif
statistics are estimated using a combination of a random walk and a
randomized waddle. 
\begin{figure}[!t]
\begin{center}
\includegraphics[width=2.5in]{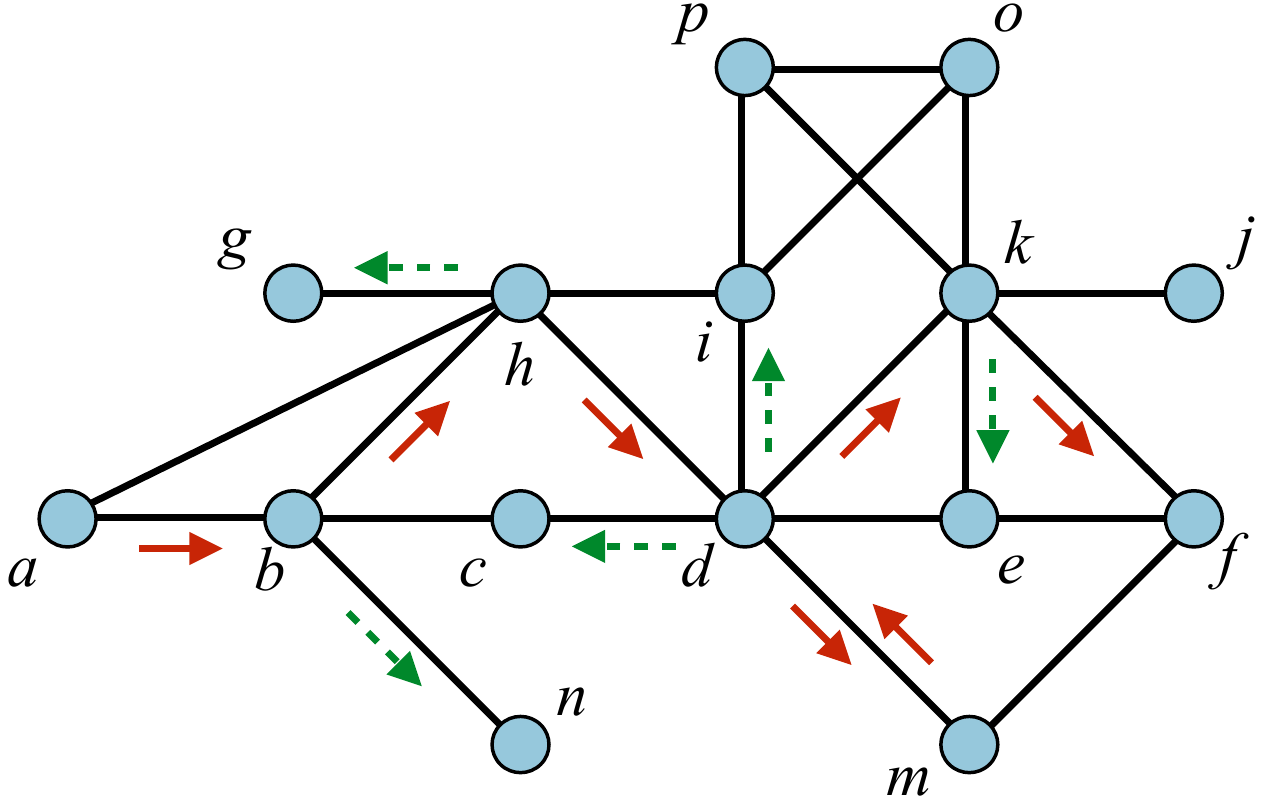}
\includegraphics[width=3.6in]{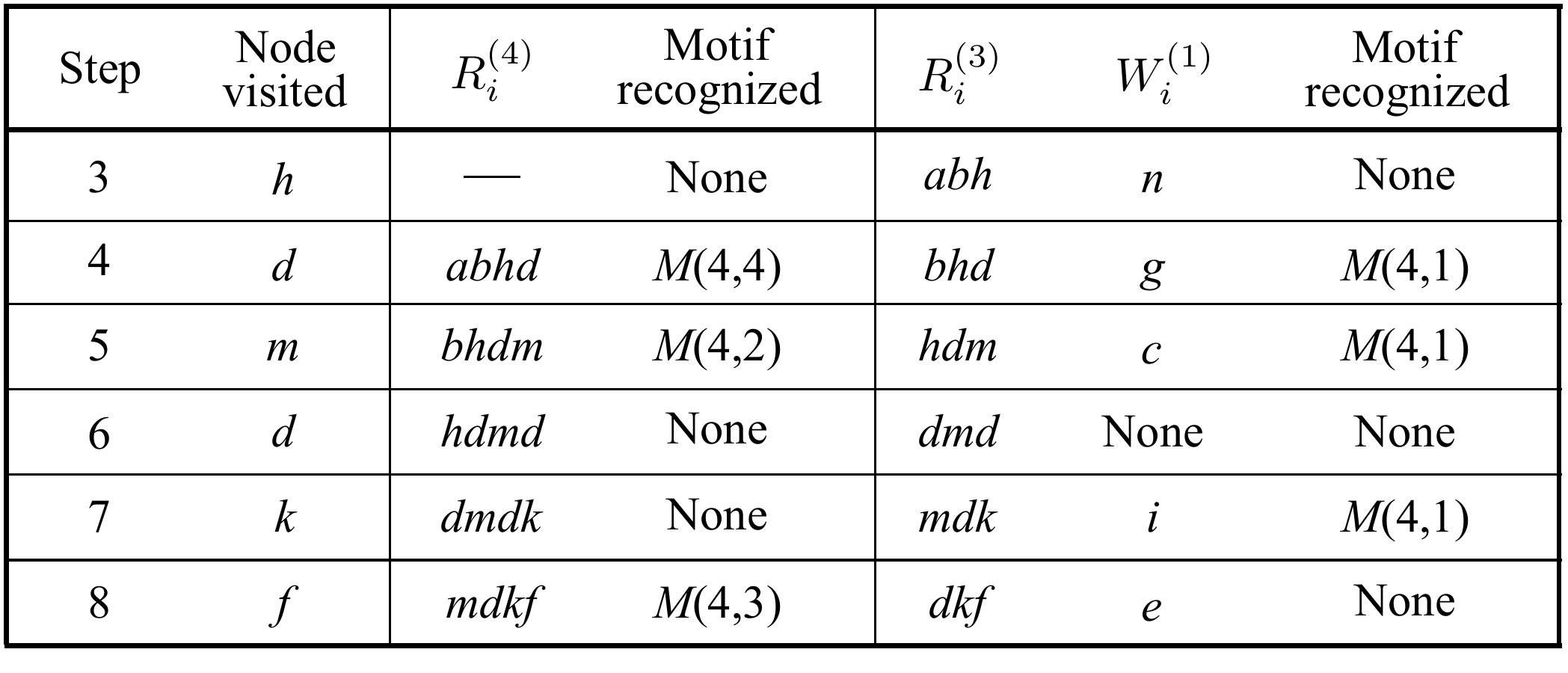} 
\caption{An example to illustrate the collection of $4$-node motif
  statistics in Waddling Random Walk. The random walk is shown by red
  arrows and the waddles are shown by dotted-line green arrows.}
\label{fig:waddle-example}
\end{center}
\end{figure}  

Lines 16--17 of the pseudocode use Eqn.~(\ref{eqn:final}) to
compute the sum $|{\bf S}(k,m)|$ by computing the LHS of the above
equation and using known values of $Z(k,m)$, $P_r(k,m,s)$ and $P_w(k,m,s)$ for
each motif. Lines 24--25 finally compute and return the motif concentration.

\subsection{Example: 4 and 5-node motif statistics}

While Algorithm~\ref{alg:motifalg} shows the pseudocode for the
Waddling Random Walk in the general case for $k$-node motifs, it is
illustrative to show how the waddle works in the case of $4$ and
$5$-node motifs. Algorithm~\ref{alg:four-node-alg} and
Algorithm~\ref{alg:five-node-alg} replace the lines 3--23 in
Algorithm~\ref{alg:motifalg} for the $4$-node case and $5$-node case,
respectively. 

\begin{algorithm}[t]
\caption{Snippet of WRW for 4-node motifs}\label{alg:four-node-alg}
\begin{algorithmic}[1]
\While{$i\leq n$}
\State{$r_i \leftarrow$ Random node in $N(r_{i-1})$}
\If{Nodes in $R_i^{(4)}$ are distinct}
\State{$m\leftarrow$ id of the motif induced by $R_i^{(4)}$}
\State{$c_m \leftarrow c_m + {\displaystyle \frac{ d(r_{i-1}) d(r_{i-2}) }{P_r(4,m,4)}}$}
\EndIf
\If{Nodes in $R_i^{(3)}$ are distinct}
\State{$w \leftarrow$ Random node in $N(r_{i-1})$}
\If{$H(R_i^{(3)} \cup w )$ is isomorphic to $M(4,1)$}
\State{$c_1 \leftarrow c_1 + {\displaystyle \frac{ d(r_{i-1})^2 }{6}}$}
\EndIf
\EndIf
\State{$i\leftarrow i+1$}
\EndWhile
\end{algorithmic}
\end{algorithm}

Lines 5 and 10 in the 4-node case are surprisingly simple compared to the
generalized case represented in line 17 of
Algorithm~\ref{alg:motifalg}. This is because $\phi(W_i^{(k-s)} |
R_i^{(s)})$ is simply $1/d(r_{i-1})$. The motif type $M(4,1)$ is the
only case in which $s=3$, for which $P_r(4,1,3) = 6$, $P_w(4,1,3) = 1$
and $Z(k,m)=1$.

In the 5-node case, only motif types $M(5,2)$, $M(5,3)$ and $M(5,6)$
have the number of vertices in their longest simple path less than
5. The corresponding values of $P_r$, $P_w$ and $Z$ are presented in
Fig.~\ref{fig:P-example}.  

\begin{algorithm}[h]
\caption{Snippet of WRW for 5-node motifs}\label{alg:five-node-alg}
\begin{algorithmic}[1]
\While{$i\leq n$}
\State{$r_i \leftarrow$ Random node in $N(r_{i-1})$}
\If{Nodes in $R_i^{(5)}$ are distinct}
\State{$m\leftarrow$ id of the motif induced by $R_i^{(5)}$}
\State{$c_m \leftarrow c_m + {\displaystyle \frac{ d(r_{i-1}) d(r_{i-2}) d(r_{i-3})}{P_r(5,m,5)}}$}
\EndIf
\If{Nodes in $R_i^{(4)}$ are distinct}
\State{$w \leftarrow$ Random node in $N(r_{i-2})$}
\If{$H(R_i^{(4)} \cup w )$ is isomorphic to $M(5,2)$}
\State{$c_2 \leftarrow c_2 + {\displaystyle \frac{ d(r_{i-1})d(r_{i-2})^2 }{2}}$}
\ElsIf{$H(R_i^{(4)} \cup\,w )$ is isomorphic to $M(5,6)$}
\State{$c_6 \leftarrow c_6 + {\displaystyle \frac{ d(r_{i-1})d(r_{i-2})^2 }{4}}$}
\EndIf
\EndIf
\If{Nodes in $R_i^{(3)}$ are distinct}
\State{$w_1 \leftarrow$ Random node in $N(r_{i-1})$}
\State{$w_2 \leftarrow$ Random node in $N(r_{i-1})$}
\State{$R_{\mathrm{temp}} \leftarrow R_i^{(3)} \cup \{w_1,w_2\}$}
\If{$H(R_{\mathrm{temp}}  )$ is isomorphic to $M(5,3)$}
\State{$c_3 \leftarrow c_3 + {\displaystyle \frac{ d(r_{i-1})^3 }{24}}$}
\EndIf
\EndIf
\State{$i\leftarrow i+1$}
\EndWhile
\end{algorithmic}
\end{algorithm}

\newtheorem{theorem}{Theorem}
\newtheorem{lemma}{Lemma}

\subsection{Theoretical bound}
We present a theoretical analysis of the number of steps required in the random walk in order to get an accurate estimate.
Let $T = T(\epsilon)$ be the $\epsilon$-mixing time of a Markov Chain on $G$ where $\epsilon$ is at most $\frac{1}{8}$. Let $Q$ be the product of the top $k$ degrees in $G$.

 \begin{lemma}
 \label{lem:bound1}
For $0< \delta<1$, there exists a constant $\xi$, such that for $t \geq \xi \frac{TDQ}{|{\bf{S}}(k,m)| \delta^2} \log{\frac{1}{\alpha}}$, we have
 \begin{eqnarray*}
\mathrm{Pr} \left[\left(1-\frac{2\delta}{1+\delta}\right)C(k,m)\leq \hat{C}(k,m)\leq \left(1+ \frac{2\delta}{1-\delta}\right)C(k,m)\right] > 1-2\alpha
\end{eqnarray*}
\end{lemma}
where $\hat{C}(k,m)$ is the estimated $C(k,m)$, the concentration of motif $M(k,m)$.
The detailed proof is presented in the appendix.

When the number of steps in the random walk $t \geq \xi \frac{TDQ}{|{\bf{S}}(k,m)| \delta^2} \log{\frac{1}{\alpha}}$, the relative error of an estimate of the concentration of motif $M(k,m)$ is at most $\frac{2\delta}{1-\delta}$ with probability greater than $1-2\alpha$ . Besides $\alpha$ and $ \delta$, the number of steps is also determined by the mixing time, the degree distribution of the graph and the number of subgraphs which are isomorphic to motif $M(k,m)$. 

\section{Performance Analysis}
\label{sec:results} 

\begin{figure*}[!t]
\centering
    \subfigure[{com-Amazon ($Q=8K$)}]{
       \includegraphics[width=2.5in]{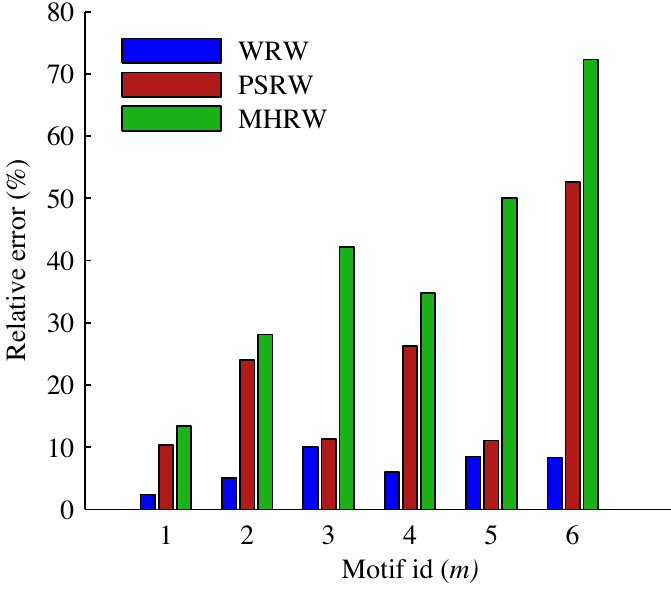}
       \label{fig:accuracy2}
       }
       \subfigure[{soc-Slashdot ($Q=8K$)}]{ 
       \includegraphics[width=2.5in]{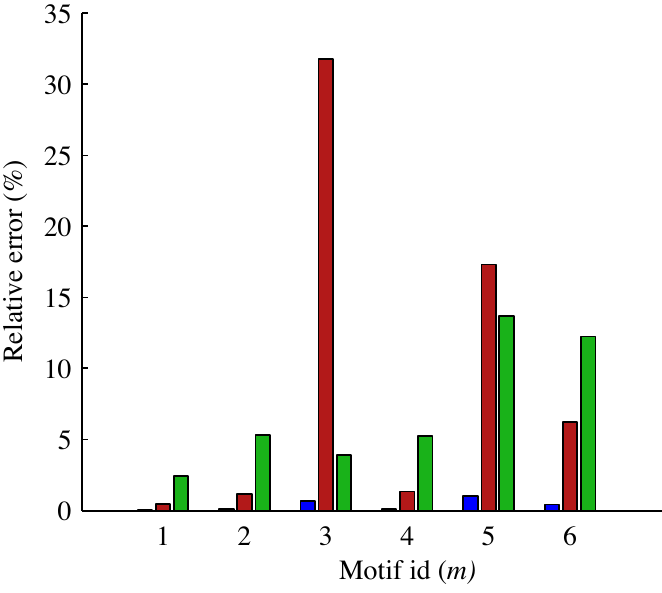}
       \label{fig:accuracy3}
       }
           \subfigure[{socfb-Penn94 ($Q=8K$)}]{
       \includegraphics[width=2.5in]{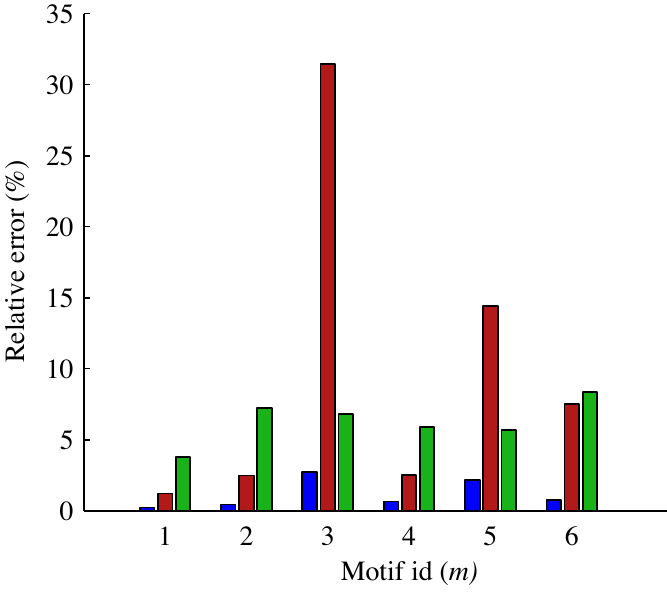}
       \label{fig:accuracy5}
       }
       \subfigure[{com-Youtube ($Q=20K$)}]{ 
       \includegraphics[width=2.5in]{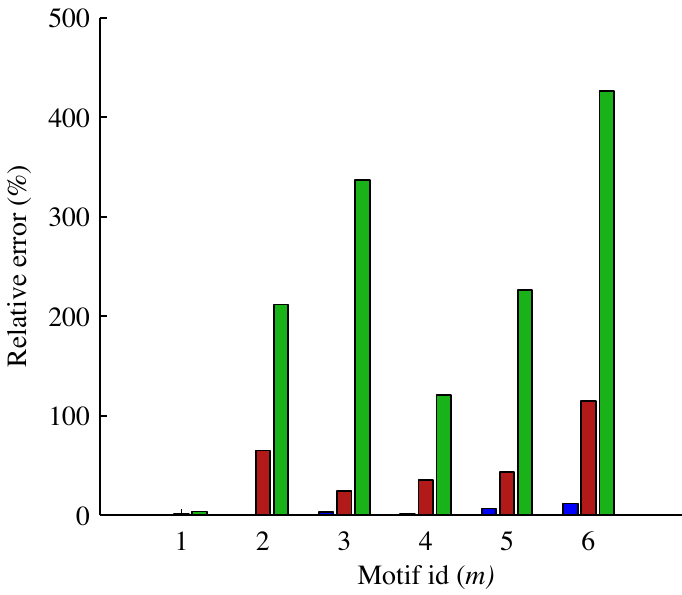}
       \label{fig:accuracy6}
       }       
   \caption{Comparison of the relative errors in the estimates of
     $4$-node motif concentrations.}\label{fig:accuracysum_four} 
\end{figure*}

\begin{figure*}[!t]
\centering

                  \subfigure[{soc-Slashdot($Q=18K$)}]{
       \includegraphics[width=4.5in]{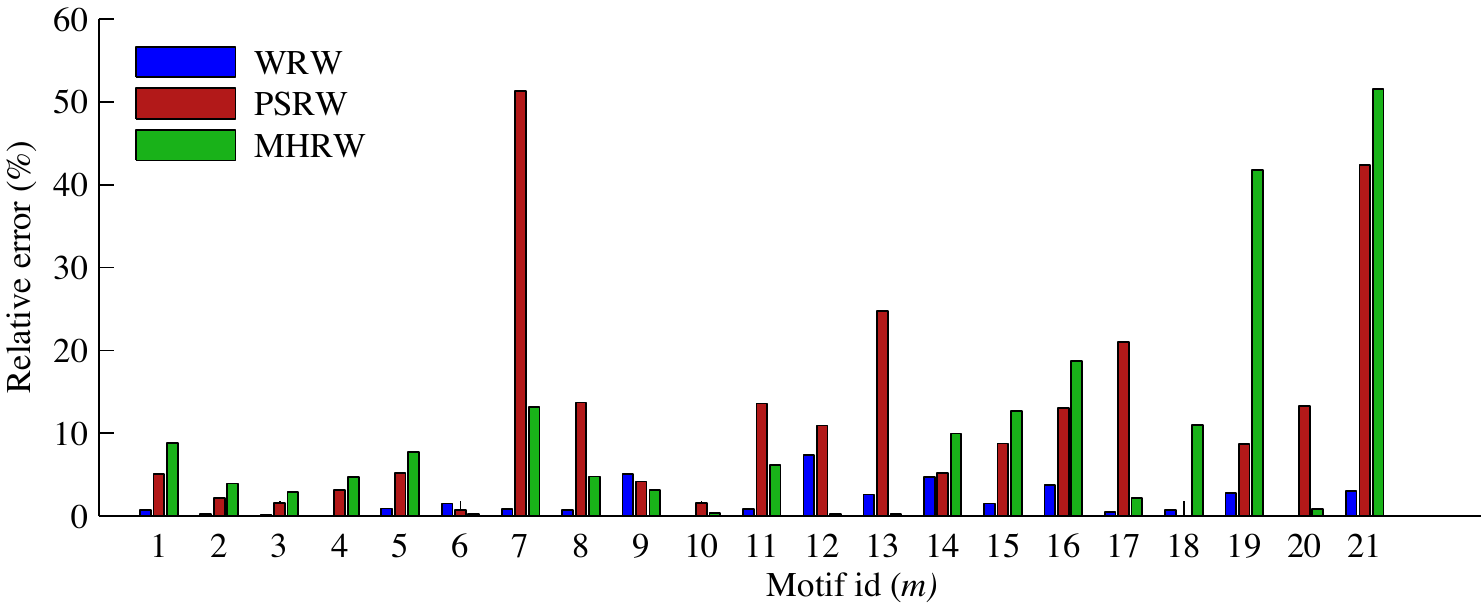}
       \label{fig:accuracy7}
       }
       \subfigure[{socfb-Penn94($Q=18K$)}]{ 
       \includegraphics[width=4.5in]{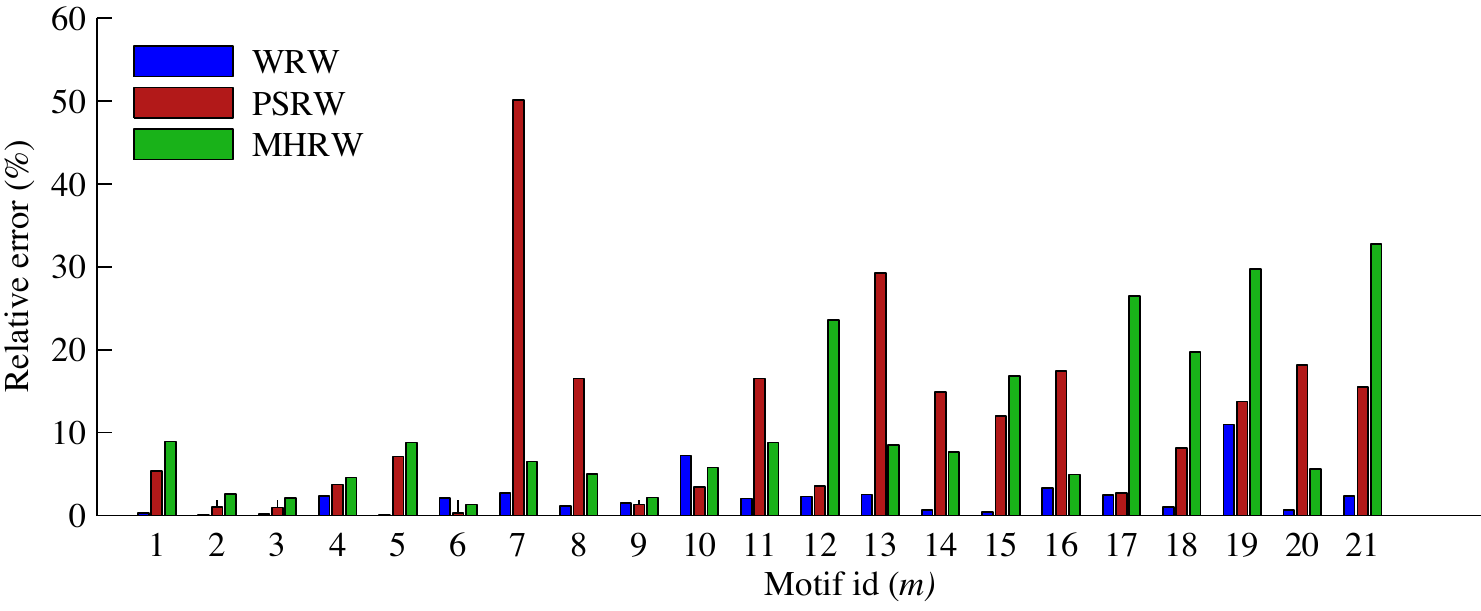}
       \label{fig:accuracy8}
       }
          \caption{Comparison of the relative errors in the estimates of
     $5$-node motif concentrations.}\label{fig:accuracysum_five} 
\end{figure*} 

In this section, we conduct a comparative performance analysis of
Waddling Random Walk (WRW) against the best two, both recently
proposed, graph sampling algorithms which address the same problem
under the same constraints on access to the full graph. The
Metropolis-Hastings Random Walk (MHRW) estimates the concentrations of
$k$-node motifs by adopting the 
Metropolis-Hastings method to perform a uniform random sampling of
connected induced subgraphs (CIS) with $k$ nodes in the large graph
\cite{SahHas2015}. The Pairwise Subgraph Random Walk (PSRW), samples a
set of CISs with $k-1$ nodes by walking on the graph of these CISs to
estimate the $k$-node motif statistics \cite{WanLui2014}. Both PSRW
and MHRW, like WRW, are capable of estimating motif concentrations of any
size. As presented in \cite{SahHas2015,WanLui2014}, these two algorithms are significantly better than previously known methods in terms of both accuracy and speed, which motivates our choice of these algorithms for the comparative analysis in this section.

\begin{table}[!b]
\centering
\caption{Graph datasets used in the analysis.\label{tab:graphdata}}{
\newcommand{\tabincell}[2]{\begin{tabular}{@{}#1@{}}#2\end{tabular}}
\centering
\begin{tabular}{c|c|c|c|c|c|cl|} 
 \Xhline{1pt}
\tabincell{c}{Graph\\(LCC)}& \tabincell{c}{Nodes\\$|V|$} &\tabincell{c}{Edges\\$|E|$}&$C(4,1)$&$C(4,6)$&$C(5,3)$&$C(5,21)$\\  \Xhline{1pt}
com-Amazon & 3.35e+05&9.26e+05& 6.99e-01&1.55e-03& 7.45e-01& 7.24e-06\\ \hline
soc-Slashdot & 7.73e+04 & 4.69e+05&6.86e-01&9.19e-05&6.15e-01 & 1.15e-06\\ \hline
socfb-Penn94 &4.15e+04 &1.36e+06 &6.52e-01  &3.59e-04&6.18e-01 & 2.30e-06\\ \hline
com-Youtube&1.13e+06&2.99e+06 &9.82e-01&   8.55e-07& ---  & ---\\
\hline\end{tabular}}
\end{table}

\begin{table}[!b]
\centering
\caption{Runtime for exact computation of motif concentrations (in seconds).\label{tab:runtime_compare}}{
\centering
\begin{tabular}{c|c|c} 
 \Xhline{1pt}
   &4-node motif&5-node motif   \\ 
Graph  &Exact computation &Exact computation \\  \Xhline{1pt}
com-Amazon & 2.14&76.98\\ \hline
soc-Slashdot & 9.21&7030.17\\ \hline
socfb-Penn94 &47.58 &178845.84  \\ \hline
com-Youtube&102.62&---\\
\hline\end{tabular}}
\end{table}



\begin{figure*}[!t]
\centering
      \subfigure[{com-Amazon: $C(4,1)/$Actual}]{ 
       \includegraphics[width=1.7in]{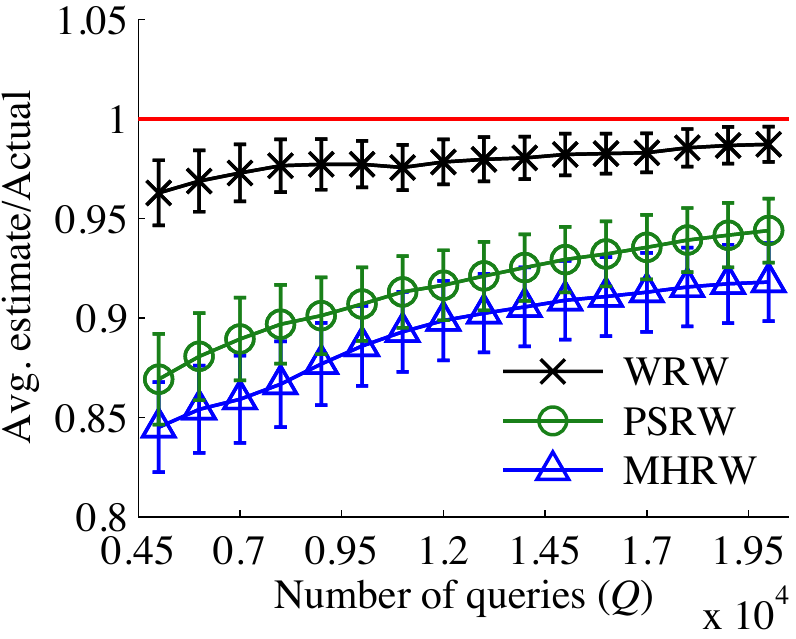}
       \label{fig:erbar_amazon_star}
       }
    \subfigure[{com-Amazon: $C(4,6)/$Actual }]{
       \includegraphics[width=1.7in]{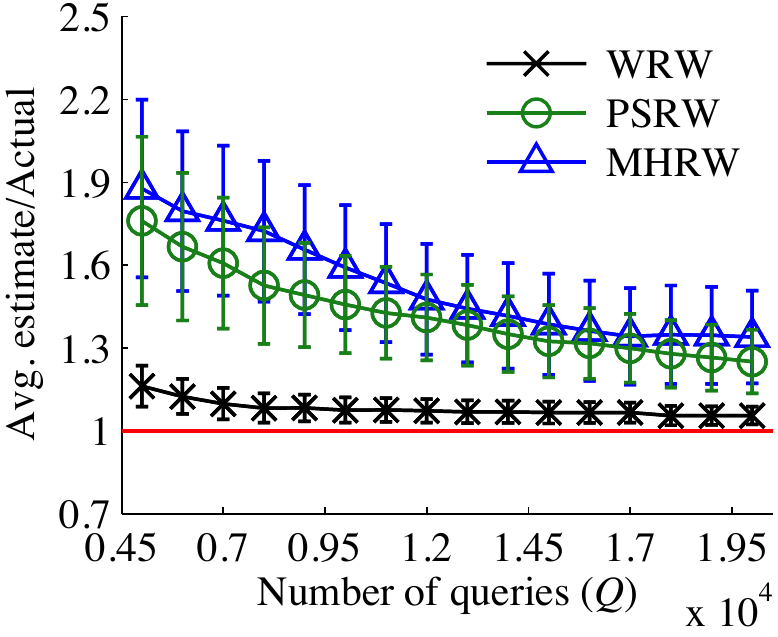}
       \label{fig:erbar_amazon_clique}
       }
       \subfigure[{soc-Slashdot: $C(4,1)/$Actual}]{ 
       \includegraphics[width=1.7in]{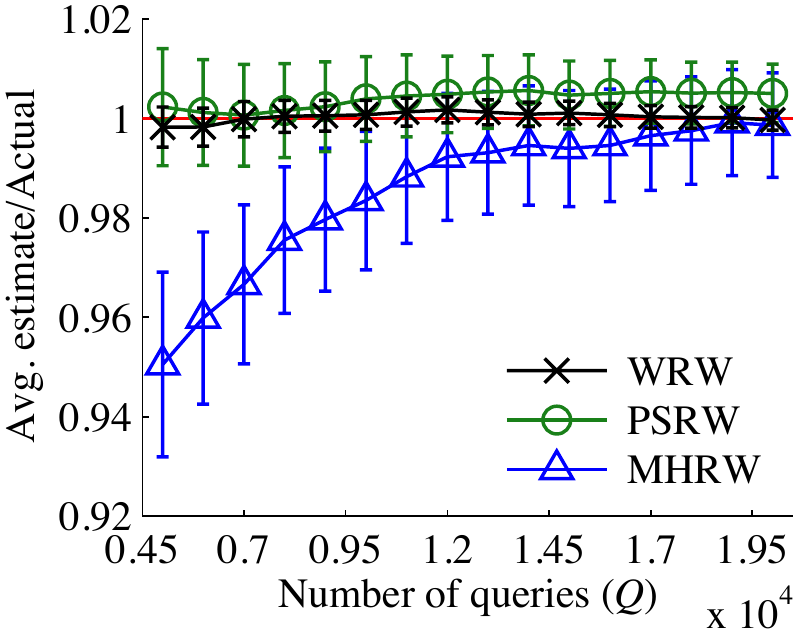}
       \label{fig:erbar_slashdot_star}
       }  
       \subfigure[{soc-Slashdot: $C(4,6)/$Actual}]{ 
       \includegraphics[width=1.7in]{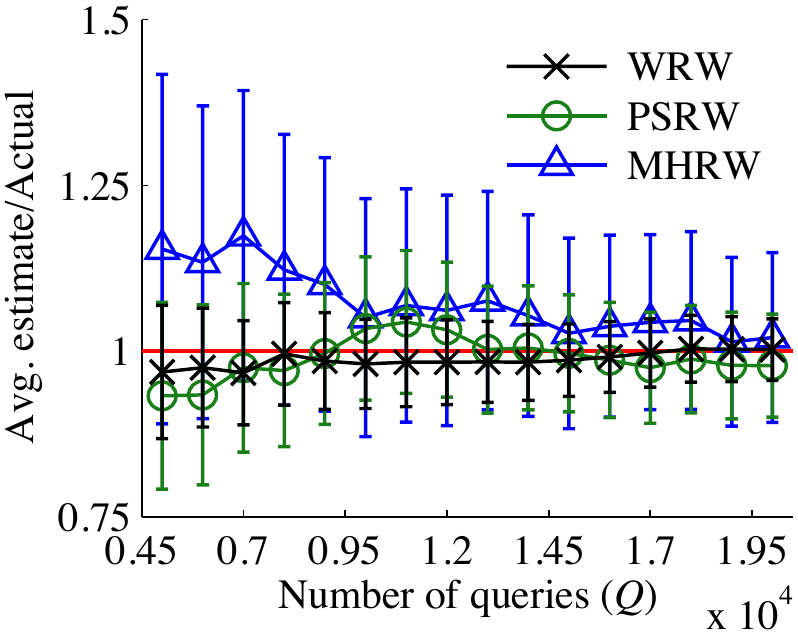}
       \label{fig:erbar_slashdot_clique}
       }
              \subfigure[{socfb-Penn94: $C(4,1)/$Actual}]{ 
       \includegraphics[width=1.7in]{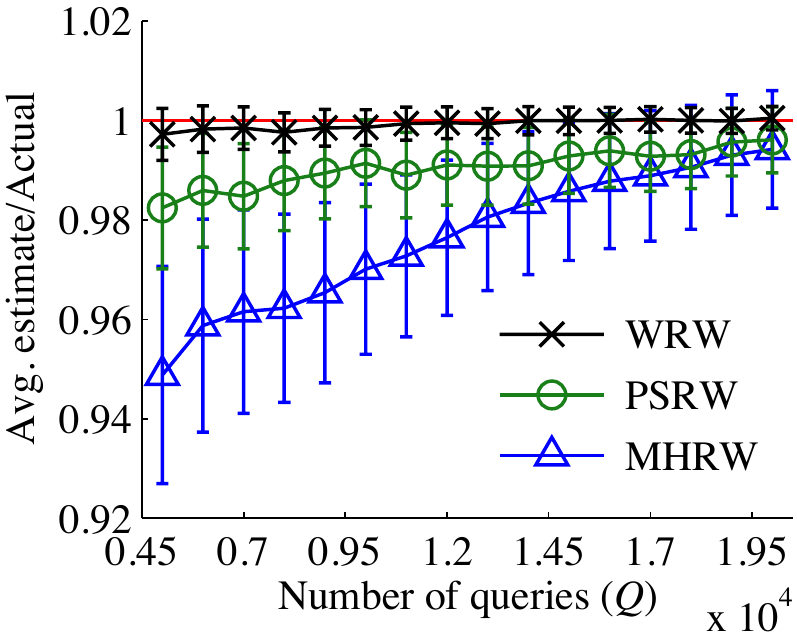}
       \label{fig:erbar_penn_star}
       }
              \subfigure[{socfb-Penn94: $C(4,6)/$Actual}]{ 
       \includegraphics[width=1.7in]{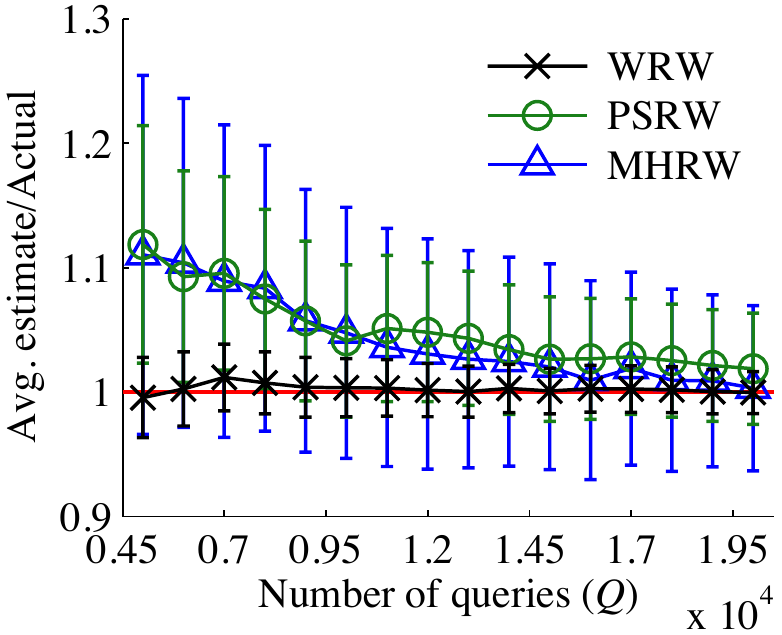}
       \label{fig:erbar_penn_clique}
       }
           \subfigure[{com-Youtube: $C(4,1)/$Actual}]{ 
       \includegraphics[width=1.7in]{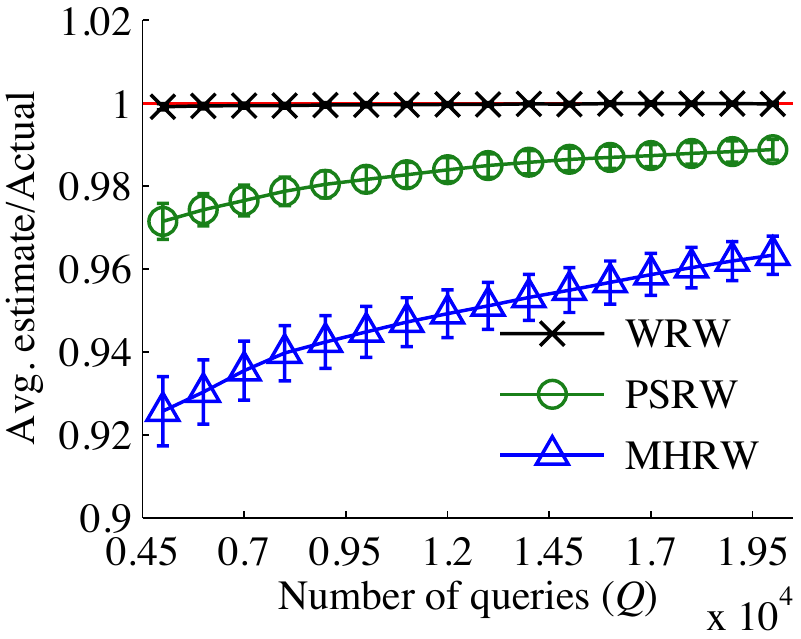}
       \label{fig:erbar_youtube_star}
       }
       \subfigure[{com-Youtube: $C(4,6)/$Actual}]{ 
       \includegraphics[width=1.7in]{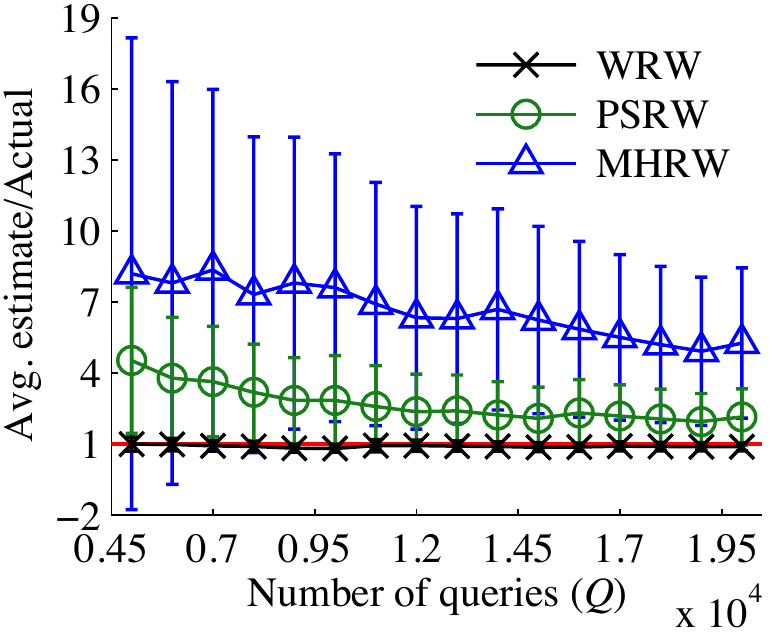}
       \label{fig:erbar_youtube_clique}
       }
       
   \caption{Comparison of the ratio of the average estimated values of
     $C(4,1)$ and $C(4,6)$ and their actual values. Red line indicates
     1. The error bars indicate 95\% confidence intervals over 200
     independent runs. }\label{fig:accuracysum2} 
\end{figure*}

\begin{figure*}[!t]
\centering
 \subfigure[{soc-Slashdot: $C(5,3)/$Actual }]{
       \includegraphics[width=2.2in]{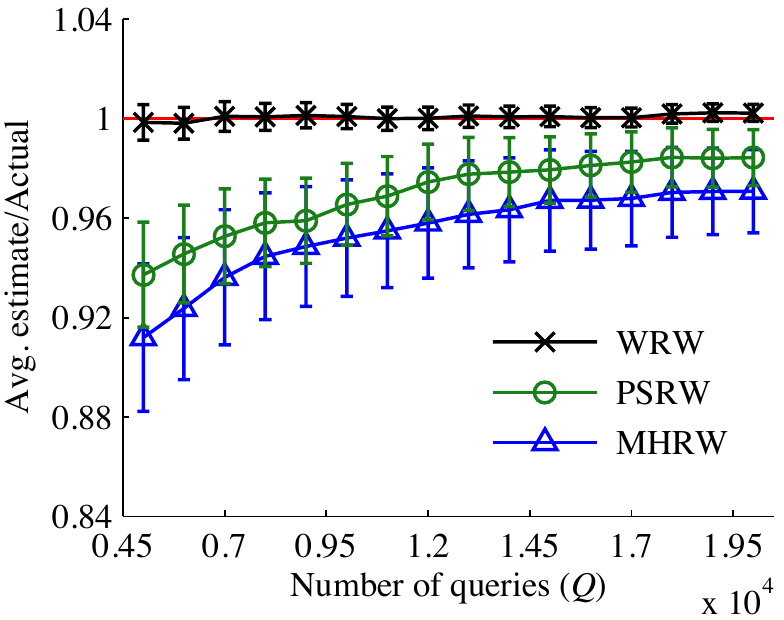}
       \label{fig:erbar_slashdot_c5_3}
       }\hskip0.05in
 \subfigure[{soc-Slashdot: $C(5,21)/$Actual}]{ 
       \includegraphics[width=2.2in]{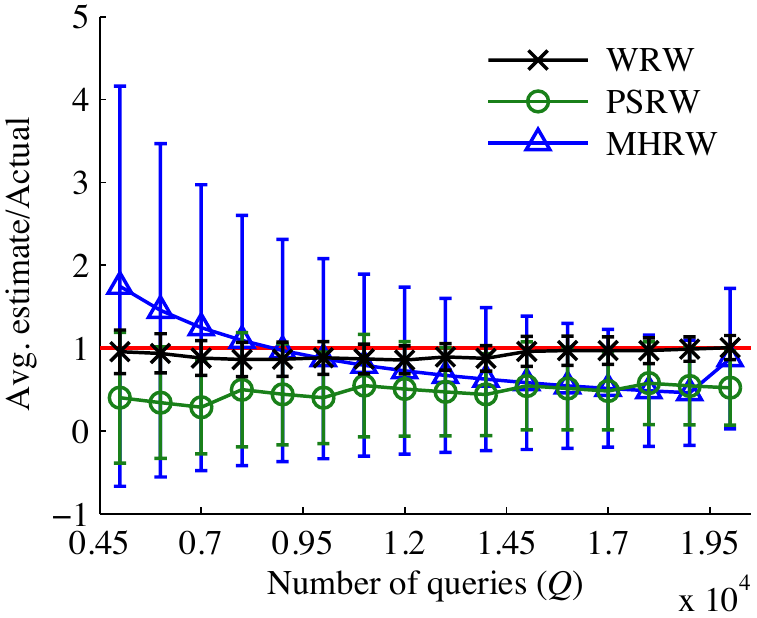}
       \label{fig:erbar_slashdot_c5_21}
       }\hskip0.05in
\subfigure[{socfb-Penn94: $C(5,3)/$Actual }]{
       \includegraphics[width=2.2in]{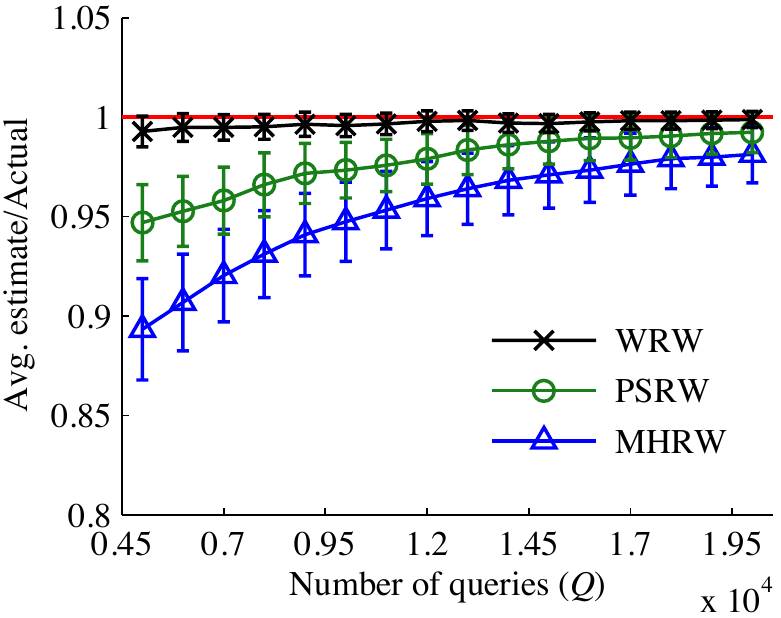}
       \label{fig:erbar_penn_c5_3}
       }\hskip0.05in
       \subfigure[{socfb-Penn94: $C(5,21)/$Actual}]{ 
       \includegraphics[width=2.2in]{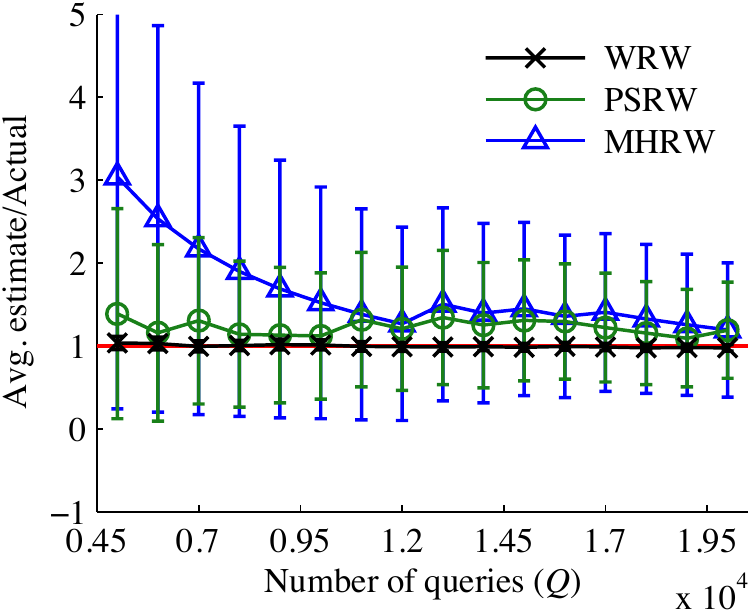}
       \label{fig:erbar_penn_c5_21}
       }
   \caption{Comparison of the ratio of the average estimated values of
     $C(5,2)$ and $C(5,3)$ and their actual values. Red line indicates
     1. The error bars indicate 95\% confidence intervals over 200
     independent runs. }\label{fig:accuracysum3} 
\end{figure*}
 
We performed our experiments on real graphs from the
Stanford Network Analysis Project (SNAP) \cite{snapnets}
 and the Koblenz Network Collection \cite{konect}. For
each graph dataset used, we run the random walk algorithms on the
largest connected component (LCC) of it. The 
name, the number of vertices and the number of edges in these graphs are
listed in Table \ref{tab:graphdata} along with the actual values 
of selected $4$-node and $5$-node motif concentrations.
Exact values of the motif concentrations are obtained by using Orbit Counting Algorithm (Orca) \cite{HocDem2014}. 
It reduces the time complexity of existing direct enumeration methods
by an order of magnitude. One objective of this work is to make large
graph analysis possible, through sampling algorithms, on ordinary
desktops and, therefore, all of the performance analysis in this paper
was conducted on an iMac with 8GB 1600MHz DDR3 memory and a 2.7GHz
Intel Core i5 processor. The runtime of using Orca to calculate
4-node and 5-node motif concentrations on the iMac is reported in
Table \ref{tab:runtime_compare}. As shown in the table, the runtime
for calculating motif concentrations increases dramatically as the
graph size increases. For example, for socfb-Penn94 graph, it takes
more than 2 days to obtain the exact concentrations of 5-node
motifs. For com-Youtube graph, the largest graph among all four
graphs, we estimate more than one month of runtime to complete the
exact computation of 5-node motif concentrations --- the data of 5-node motifs
in com-Youtube graph, therefore, is not provided here.

We conduct our comparative analysis based on four key metrics: the
number of queries, the run time, the accuracy (how close is the
estimate to the correct answer?), and the precision (how low is the
variance in the estimates?). The first two address the speed of the
algorithms and the latter two address the confidence we should have in
the estimates. If the amount of processing done per query by
the algorithms are different, the number of queries is not directly
indicative of the speed of the algorithm --- therefore, we also use the actual
run time of the algorithms in our analysis. The number of
queries, however, is still meaningful as a metric since it indicates
the amount of information collected by the algorithm.

\subsection{Accuracy and precision}
\label{subsec:acc}
Figs.~\ref{fig:accuracysum_four} and \ref{fig:accuracysum_five} show relative errors in estimating the
concentrations of each of the $4$-node and 5-node motifs for each of the three
algorithms. We measure the relative error as:
\[
\mathrm{Relative~error} = \frac{\mathrm{Average~estimate} -
  \mathrm{Actual~value}}{\mathrm{Actual~value}}
\]
The average estimate is calculated as the mean of the estimated value
over 200 independent runs. 

For each graph, we fixed $Q$, the number of queries.
For PSRW and MHRW, since the CIS sampled in the next step differs from the CIS of
the current step in only one node, we assume only one query per CIS
considered. As shown in Figs.~\ref{fig:accuracysum_four} and  \ref{fig:accuracysum_five},
for almost all motif types, the WRW algorithm proposed in this paper yields a smaller
relative error, i.e., higher accuracy, for the same number of queries
than the other two algorithms, and especially so when the actual concentration of the motif is low. 


Figs.~\ref{fig:accuracysum2} and \ref{fig:accuracysum3} help evaluate both the accuracy and the
precision of the WRW algorithm in comparison to PSRW and MHRW in
4-node case and 5-node case. They show the ratio of the average
estimated motif concentration to the actual value for each of the four
graphs with increasing number of queries. 

In large graphs, the star motifs, e.g., $M(4,1)$ and $M(5,3)$,
typically have the largest concentration. The clique motifs, such
as $M(4,6)$ and $M(5,21)$, have the smallest concentration in most
cases. In order to demonstrate accuracy and precision spanning
the full range of actual motif concentrations, we choose to plot the
estimates for $C(4,1)$ and $C(4,6)$ in case of 4-node motifs (in
Fig.~\ref{fig:accuracysum2}) and $C(5,3)$ and $C(5,21)$ in case of
5-node motifs (in Fig.~\ref{fig:accuracysum3}). 
As shown in Fig.~\ref{fig:accuracysum2} and \ref{fig:accuracysum3}, to obtain similar accuracy and precision, fewer nodes are queried for the star motifs  than the clique motifs. When the 
concentration is high, it is easier to reach good accuracy and
precision since the random walk will encounter more samples. But, note
that, as also shown in Figs.~\ref{fig:accuracysum_four} and
\ref{fig:accuracysum_five}, the WRW algorithm achieves especially good
accuracy and precision when doing so is harder, i.e., when the
motif concentration is very low.

The closeness of the WRW plot to the red line indicates its
significantly better accuracy than the other algorithms. Also, the smaller error bars on
the WRW plot show that, besides being more accurate, the estimates
made by WRW are also more precise compared to the other algorithms.






\subsection{Runtime}

The WRW algorithm achieves an improvement in the runtime by avoiding
the enumeration of subgraphs but instead simply picking a random set of
nodes and checking for isomorphism. This makes a particular difference
in the case of graphs with high average node-degree.

We implemented the three algorithms, WRW, PSRW and MHRW, in Python using
{\em iGraph} routines. Since part of the point of graph sampling is to make
Big Data analysis feasible on ordinary desktops, we ran all of the
simulations on the same ordinary iMac as that used for
Table~\ref{tab:runtime_compare}. We assume that the graph datasets are stored on the
local machine and, so, the time lost to querying corresponds to the
time involved in accessing the memory.

Fig.~\ref{fig:runtime} plots the relative error of the estimates of
$C(4,1)$ and $C(5,3)$ made by the three algorithms against the total runtime,
averaged over 100 independent runs on com-Youtube and
socfb-Penn94 graphs. The figure shows that WRW achieves significantly
better accuracy for the same runtime than other algorithms in both 4-node and 5-node motif cases. 

\begin{figure}[!t]
\centering
\subfigure[{com-Youtube: $C(4,1)$}]{
       \includegraphics[width=2in]{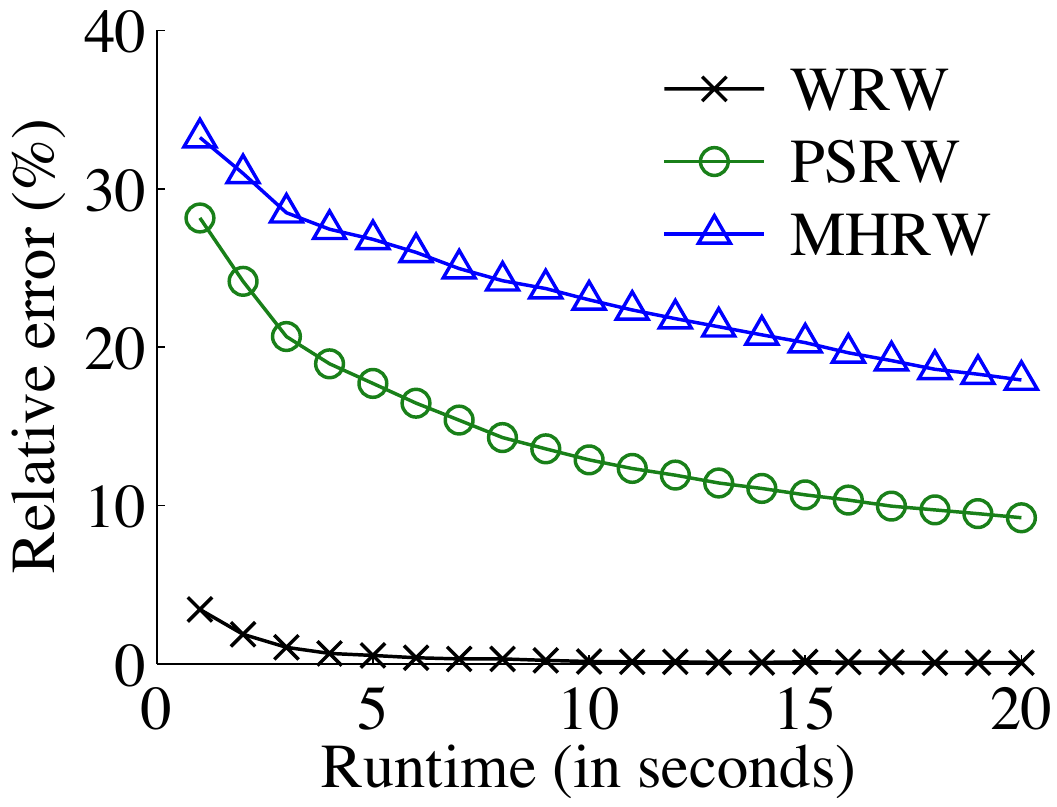}
       \label{fig:runtime2}
       }
       \subfigure[{socfb-Penn94: $C(5,3)$}]{
       \includegraphics[width=2in]{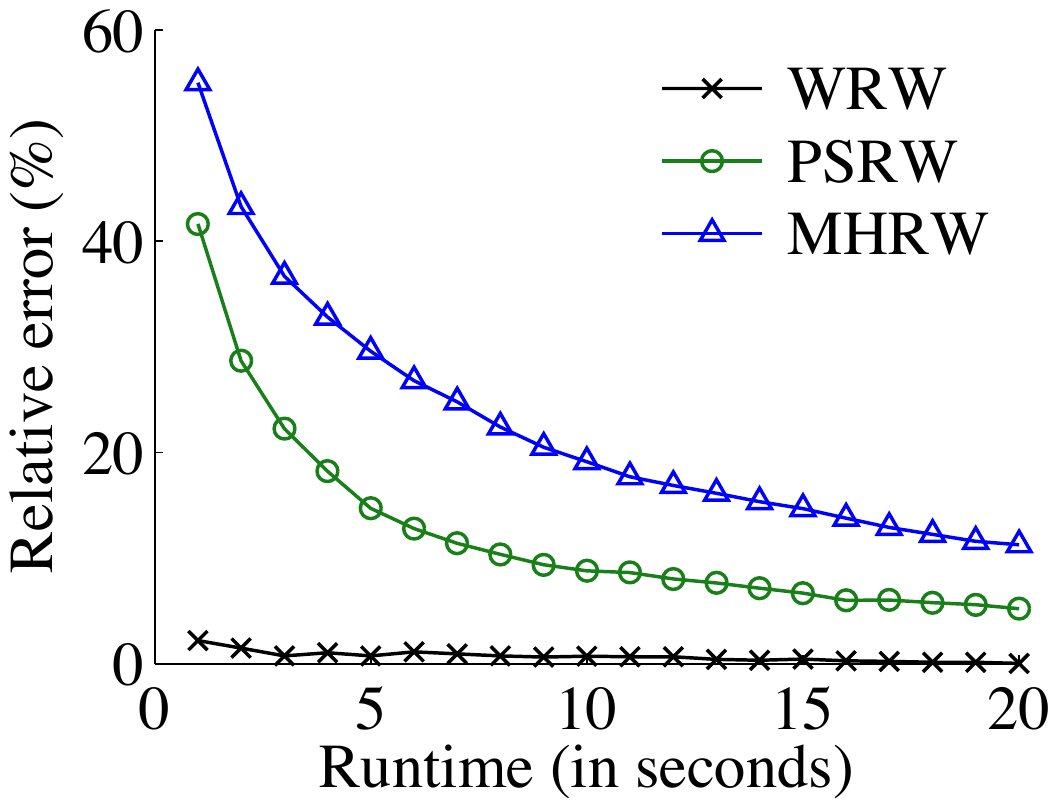}
       \label{fig:runtime3}
       }

\caption{Comparisons of the relative error in estimating $C(4,1)$ and
  $C(5,3)$ against the runtime in seconds.}
\label{fig:runtime}
\end{figure}

\subsection{Application to motif counting}
\label{sec:application}

Motif analysis has been used as an important technique for analyzing complex networks \cite{ChuKwa2008}. Networks with identical global graph properties, such as global clustering coefficient and diameter, may have different local structures \cite{MilShe2002}. The number of occurrences of each motif has been widely used for uncovering the local structure of the networks.  \cite{MilPrz2008} summarized a protein's local structure in a protein-protein interaction (PPI) network via the motif degree signature.
 \cite{SheVis2009} proposed efficient graph kernels based on counting or sampling motifs to characterize and compare graphs. In \cite{HalArt2008}, motif counts were used to build subgraph ratio profiles for comparing P2P networks with protein structure networks.

Our algorithm can also be applied to estimating the motif counts with the help of the network size estimator in \cite{HarKat2013}. 
As presented in Eqn. (\ref{eqn:final}), the number of subgraphs which are isomorphic to motif $M(k,m)$ equals the ratio of $ E[\omega(R_i^{(s)} \cup W_i^{(s)}, k, m)
  \frac{\prod_{j=1}^{s-2}d(r_{i-j})}{\phi(W_i^{(k-s)} | R_i^{(s)})}
  ]$ and $\left( \frac{1}{D} \right) \left(
  \frac{P_r(k,m,s) P_w(k,m,s)}{Z(k,m)} \right)$. The former can be estimated by running the WRW algorithm. $P_r(k,m,s)$, $P_w(k,m,s)$ and $Z(k,m)$ are known for all the motifs. $D$ is the sum of the degrees of all the nodes in $G$. We assume that the entire graph is not accessible, so the real value of $D$ remains unknown; however, we can estimate the value of $D$ via random walk.
  
Let $x_i  \in V$ denote the node visited in step $i$. Consider the expected value of $\frac{1}{d(x_i)} $ over the random walk:
  \begin{eqnarray}
E\left[\frac{1}{d(x_i)}\right] &=& \sum_{v \in V} \frac{d(v)}{D}\frac{1}{d(v)} =\frac{|V|}{D}
\label{eqn:appl-motif}
\end{eqnarray}
Eqn. (\ref{eqn:appl-motif}) shows that $D$ is equal to the ratio of the number of nodes in the entire graph and the expected value of $\frac{1}{d(x_i)}$. \cite{HarKat2013} presents a sampling method which estimates the number of nodes by counting the neighbor collision of node pairs in the random walk. 

Having the estimated value of $D$, WRW is capable of estimating the number of motifs of any type. Given a $k$-node motif $M(k,m)$, WRW can estimate the number of subgraphs which are isomorphic to $M(k,m)$ without performing any estimation of other types of $k$-node motifs. Other motif counting algorithms, such as \cite{JhaSes2015}, require the estimates of all non-star-like motifs in order to estimate the motif count of a star-like motif.

Table \ref{tab:D_est} reports the accuracy of the estimates of $D$, and Table \ref{tab:motif_counts} shows the relative errors in the estimates of selected 4,5-node motif counts. As mentioned in Section \ref{subsec:acc}, in most of the cases, the star motifs  have the largest concentration, while the clique motifs have the smallest concentration. So we choose to present the estimates for $|\bf{S}(4,1)|$ and $|\bf{S}(4, 6)|$ in case of 4-node motifs, and $|\bf{S}(5, 3)|$ and $|\bf{S}(5, 21)|$ in case of 5-node motifs. As presented in Table \ref{tab:motif_counts}, even with only hundreds of nodes queried, our method still has a relative error lower than 10\%.
 
\begin{table}[!t]
\centering
\caption{The relative errors in the estimates of the sum of degrees of all nodes ($D$). Only 1\% of the nodes is mined.\label{tab:D_est}}{
\centering
\begin{tabular}{c|c|c|c} 
 \Xhline{1pt}
Graph  &Sample size & $D$& Relative error (\%) \\  \Xhline{1pt}
soc-Slashdot & 774&938,360&0.5753\\ \hline
socfb-Penn94 &415 &2,724,440& 0.6748 \\ 
\hline\end{tabular}}
\end{table}

\begin{table}[!t]
\centering
\caption{The relative errors in the estimates of the number of selected 4,5-node motifs when 1\% of the nodes is mined.\label{tab:motif_counts}}{
\begin{tabular}{c|c|c|c|c} 
 \Xhline{1pt}
 \multicolumn{5}{c}{ 4-node motifs}\\  \Xhline{1pt}
Graph  & $|\bf{S}(4,1)|$& Relative error (\%)& $|\bf{S}(4,6)|$& Relative error (\%) \\   \Xhline{1pt}
soc-Slashdot &1.49e+10&3.9656&1.99e+06&4.1792\\ \hline
socfb-Penn94 &5.69e+10&4.0676&3.13e+07&0.3927\\ \Xhline{1pt}
 \multicolumn{5}{c}{ 5-node motifs}\\  \Xhline{1pt}
  Graph & $|\bf{S}(5,3)|$& Relative error (\%)& $|\bf{S}(5,21)|$& Relative error (\%) \\   \Xhline{1pt}
soc-Slashdot &5.73e+12&5.6742&1.07e+07& 9.2719\\ \hline
socfb-Penn94 &3.86e+13& 0.6748&1.44e+08&8.0219 \\ 
\hline\end{tabular}
}
\end{table}

\section{Conclusions}
\label{sec:conclusion}
This paper demonstrates a simple approach, based on a random walk, to
collect and estimate the motif statistics of a large graph by sampling
only a small fraction of the motifs in the graph. The algorithm,
called Waddling Random Walk, is significantly faster than other known algorithms
that address both the computational and access challenges in the
subgraph mining of large graphs. 
The key feature of WRW that
contributes to its speed is the fact that it avoids any enumeration of
subgraphs, relying instead on a randomized protocol to sample
subgraphs. Further, WRW also avoids a dependence on the node degree in
the critical path of its computations and therefore, performs
particularly well on graphs with large average node-degree.
Besides
the improvement in speed, the algorithm is also more accurate (its
estimates are closer to the correct answer) and more precise (its
estimates have low variance) than the best of previously known
algorithms. 

A powerful feature of our methodology is that it also offers a
generalized approach which can be customized to optimize for specific
motifs of interest. In fact, the theoretical rationale used for our
approach only requires that our method of waddling used in the
algorithm be a randomized protocol. Within this approach, there is
much room for improving the waddling protocol and we hope this paper
will form the foundation for new approaches leading to newer and
better algorithms in the computationally feasible analysis of large graphs.

\appendix
\setcounter{section}{1}
In this section, we present the proof of Lemma \ref{lem:bound1}.

 \begin{theorem}
 \label{thm:bound}
 (Theorem 3 \cite{ChuLam2012}) Let $M$ be an ergodic Markov chain with state space $[n]$ and stationary distribution
$\pi$. Let $T = T(\epsilon)$ be its $\epsilon$-mixing time for $\epsilon \leq 1/8$. Let $(x_1, . . . , x_t)$ denote a $t$-step random walk on
$M$ starting from an initial distribution $\varphi$ on $[n]$. For every $i \in [t]$, let $f_i
: [n] \rightarrow [0, 1]$ 
be a weight function at step $i$ such that the expected weight $E_{x\leftarrow \pi}[f_i(x)] = \mu $ for all $i$. Define the
total weight of the walk $(x_1, . . . , x_t)$ by $X \triangleq\sum_{i=1}^t f_i(x_i)$. There exists some constant $c$ (which is
independent of $\mu$, $\epsilon$ and $\delta$) such that
\begin{eqnarray*}
\mathrm{Pr} \left[  \left|\frac{X}{t}-\mu \right| > \delta \mu \right] \leq c||\varphi||_\pi e^{-\delta^2\mu t/(72 T)}
\end{eqnarray*}
where  $0< \delta<1$.
\end{theorem}
The above theorem, presented in \cite{ChuLam2012}, provides the theoretical foundation of our proof.

\begin{lemma}
\label{lem:bound}
There exists a constant $\xi$, such that for $t \geq \xi \frac{TDQ}{|{\bf{S}}(k,m)| \delta^2} \log{\frac{1}{\alpha}}$, we have
\begin{eqnarray*}
\mathrm{Pr} \left[ (1- \delta) \frac{|{\bf{S}}(k,m)|}{D} \leq\frac{c_m}{t}\leq (1+\delta) \frac{|{\bf{S}}(k,m)|}{D} \right] >1- \alpha
\end{eqnarray*}
\end{lemma}

\begin{proof} 
Define $Q$ as the product of the top $k$ degrees in $ G$. Let
\[
f_i=\frac{1}{Q}[\omega(R_i^{(s)} \cup W_i^{(s)}, k, m)\frac{\prod_{j=1}^{s-2}d(r_{i-j})}{\phi(W_i^{(k-s)} | R_i^{(s)})}]
  \left(
  \frac{Z(k,m)}{P_r(k,m,s) P_w(k,m,s)} \right).
  \]
This function comes from our estimator in Eqn. (\ref{eqn:final}). Suppose that the random walk starts from a stationary distribution $\pi$, and thus we have $||\varphi||_\pi =1$. The expected value $E[f_i]  = \frac{|{\bf{S}}(k,m)|}{DQ}$.
 Applying Theorem \ref{thm:bound},
 \begin{eqnarray*}
&\mathrm{Pr} \left[  \left|\frac{c_m}{tQ}-\frac{|{\bf{S}}(k,m)|}{DQ} \right| > \delta \frac{|{\bf{S}}(k,m)|}{DQ} \right] \\ 
=& ~~~~ \mathrm{Pr} \left[  \left|\frac{c_m}{t}-\frac{|{\bf{S}}(k,m)|}{D} \right| > \delta \frac{|{\bf{S}}(k,m)|}{D} \right]
\leq ce^{-\delta^2\frac{|{\bf{S}}(k,m)|t}{72DQT}} 
\end{eqnarray*}
As a reminder, we use $T_k$ different temporary variables, $c_m$ for $ 1 \leq m\leq T_k$, to record the $T_k$ motif concentrations (see Algorithm~\ref{alg:motifalg}). Thus, we have 
 \begin{eqnarray*}
 c_m=\sum_{i=1}^t f_i
  \end{eqnarray*}
Taking the expectation of $c_m$, 
 \begin{eqnarray*}
E\left[\frac{c_m}{t}\right]&=& E\left[\omega(R_i^{(s)} \cup W_i^{(s)}, k, m)
  \frac{\prod_{j=1}^{s-2}d(r_{i-j})}{\phi(W_i^{(k-s)} | R_i^{(s)})}
  \right]  \left(
  \frac{Z(k,m)}{P_r(k,m,s) P_w(k,m,s)} \right) \\
  &=&\frac{|{\bf{S}}(k,m)|}{D}
  \end{eqnarray*}
Let $\alpha = ce^{-\delta^2\frac{|{\bf{S}}(k,m)|t}{72DQT}}$, and thus we have $t \geq \xi \frac{TDQ}{|{\bf{S}}(k,m)| \delta^2} \log{\frac{1}{\alpha}}$.
 \end{proof}
 
 \begin{lemma}
 \label{lem:bound2}
 There exists a constant $\xi$, such that for $t \geq \xi \frac{TDQ}{\sum_{m=1}^{T_k}|{\bf{S}}(k,m)| \delta^2} \log{\frac{1}{\alpha}}$, we have
\begin{eqnarray*}
\mathrm{Pr} \left[ (1-\delta) \sum_{m=1}^{T_k}\frac{|{\bf{S}}(k,m)|}{D} \leq \frac{c_t}{t}\leq (1+ \delta) \sum_{m=1}^{T_k}\frac{|{\bf{S}}(k,m)|}{D} \right] > 1-\alpha
\end{eqnarray*}
\end{lemma}
Note $c_t = \sum^{T_k}_{j=1} {c_j}$.
\begin{proof} 
The proof is similar to the one in Lemma \ref{lem:bound}. 
\end{proof}

Applying Lemma \ref{lem:bound} and \ref{lem:bound2}, we can find that when the number of steps  $t \geq \xi \frac{TDQ}{|{\bf{S}}(k,m)| \delta^2} \log{\frac{1}{\alpha}}$,  the relative errors of the estimates of $\frac{|{\bf{S}}(k,m)|}{D}$ and $\sum_{m=1}^{T_k}\frac{|{\bf{S}}(k,m)|}{D}$ are at most $\delta$ with probability greater than  $1-\alpha$. Thus, there exists a constant $\xi$, such that for $t \geq \xi \frac{TDQ}{|{\bf{S}}(k,m)| \delta^2} \log{\frac{1}{\alpha}}$, we have

\begin{eqnarray*}
\mathrm{Pr} \left[\left( \frac{1-\delta}{1+\delta}\right) \frac{|{\bf{S}}(k,m)|}{\sum_{m=1}^{T_k}|{\bf{S}}(k,m)|} \leq \frac{c_m}{c_t}\leq \left( \frac{1+\delta}{1-\delta}\right) \frac{|{\bf{S}}(k,m)|}{\sum_{m=1}^{T_k}|{\bf{S}}(k,m)|} \right] > 1-2\alpha
\end{eqnarray*}
Noting that $C(k,m)=|{\bf{S}}(k,m)|/(\sum_{m=1}^{T_k}|{\bf{S}}(k,m)|)$,
\begin{eqnarray*}
\mathrm{Pr} \left[\left(1- \frac{2\delta}{1+\delta}\right)C(k,m) \leq \frac{c_m}{c_t}\leq \left(1+ \frac{2\delta}{1-\delta}\right)C(k,m)\right] > 1-2\alpha
\end{eqnarray*}
This proves Lemma \ref{lem:bound1}.



\bibliographystyle{IEEETran}
\bibliography{refs}
\end{document}